\documentclass[conference]{IEEEtran}
\IEEEoverridecommandlockouts

\usepackage{cite}
\usepackage{amsmath,amssymb,amsfonts}
\usepackage{algorithmic}
\usepackage{graphicx}
\usepackage{textcomp}
\usepackage{xcolor}

\usepackage[mathscr]{eucal}
\usepackage{epsfig,epsf,psfrag}
\usepackage{amssymb,amsthm,amsfonts,latexsym}
\usepackage{xcolor,url,overpic}
\usepackage{caption}
\usepackage{subcaption}
\usepackage{array}
\usepackage{verbatim}
\usepackage{bm}
\usepackage{bbm}
\usepackage{dsfont}
\usepackage{algorithmic}
\usepackage{algorithm}
\usepackage{verbatim}
\usepackage{textcomp}
\usepackage{mathrsfs}
\usepackage{epstopdf}
\usepackage{booktabs} 
\usepackage{thm-restate}
\usepackage{tikz}

\newcommand{\openone}{\leavevmode\hbox{\small1\normalsize\kern-.33em1}}

\catcode`~=11 \def\UrlSpecials{\do\~{\kern -.15em\lower .7ex\hbox{~}\kern .04em}} \catcode`~=13 

\allowdisplaybreaks[1]

\newcommand{\calF}{\mathcal{F}}

\newcommand{\calP}{\mathcal{P}}

\newcommand{\calX}{\mathcal{X}}
\newcommand{\calY}{\mathcal{Y}}
\newcommand{\calZ}{\mathcal{Z}}

\newcommand{\bbE}{\mathbb{E}}

\newcommand{\bbP}{\mathbb{P}}

\newcommand{\bbR}{\mathbb{R}}

\DeclareMathAlphabet{\mathbsf}{OT1}{cmss}{bx}{n}
\DeclareMathAlphabet{\mathssf}{OT1}{cmss}{m}{sl}

\DeclareSymbolFont{bsfletters}{OT1}{cmss}{bx}{n}  
\DeclareSymbolFont{ssfletters}{OT1}{cmss}{m}{n}
\DeclareMathSymbol{\bsfGamma}{0}{bsfletters}{'000}
\DeclareMathSymbol{\ssfGamma}{0}{ssfletters}{'000}
\DeclareMathSymbol{\bsfDelta}{0}{bsfletters}{'001}
\DeclareMathSymbol{\ssfDelta}{0}{ssfletters}{'001}
\DeclareMathSymbol{\bsfTheta}{0}{bsfletters}{'002}
\DeclareMathSymbol{\ssfTheta}{0}{ssfletters}{'002}
\DeclareMathSymbol{\bsfLambda}{0}{bsfletters}{'003}
\DeclareMathSymbol{\ssfLambda}{0}{ssfletters}{'003}
\DeclareMathSymbol{\bsfXi}{0}{bsfletters}{'004}
\DeclareMathSymbol{\ssfXi}{0}{ssfletters}{'004}
\DeclareMathSymbol{\bsfPi}{0}{bsfletters}{'005}
\DeclareMathSymbol{\ssfPi}{0}{ssfletters}{'005}
\DeclareMathSymbol{\bsfSigma}{0}{bsfletters}{'006}
\DeclareMathSymbol{\ssfSigma}{0}{ssfletters}{'006}
\DeclareMathSymbol{\bsfUpsilon}{0}{bsfletters}{'007}
\DeclareMathSymbol{\ssfUpsilon}{0}{ssfletters}{'007}
\DeclareMathSymbol{\bsfPhi}{0}{bsfletters}{'010}
\DeclareMathSymbol{\ssfPhi}{0}{ssfletters}{'010}
\DeclareMathSymbol{\bsfPsi}{0}{bsfletters}{'011}
\DeclareMathSymbol{\ssfPsi}{0}{ssfletters}{'011}
\DeclareMathSymbol{\bsfOmega}{0}{bsfletters}{'012}
\DeclareMathSymbol{\ssfOmega}{0}{ssfletters}{'012}

\newcommand{\hatD}{\hat{D}}

\newcommand{\hatY}{\hat{Y}}

\usepackage{thmtools, thm-restate}
\newtheorem{theorem}{Theorem} 
\newtheorem{lemma}[theorem]{Lemma}

\newtheorem{proposition}[theorem]{Proposition}
\newtheorem{corollary}[theorem]{Corollary}
\newtheorem{definition}{Definition}

\newtheorem{remark}{Remark}

\newcommand{\E}{\mathbb E}
\newcommand{\KL}{\mathrm{KL}}
\newcommand{\op}{\mathrm{op}}
\newcommand{\gen}{\mathrm{gen}}
\usepackage{mathtools}
\mathtoolsset{showonlyrefs}

\def\BibTeX{{\rm B\kern-.05em{\sc i\kern-.025em b}\kern-.08em
    T\kern-.1667em\lower.7ex\hbox{E}\kern-.125emX}}
\begin{document}

\title{On the Generalization of Knowledge Distillation: An Information-Theoretic View}

\author{
\IEEEauthorblockN{Bingying Li\qquad  Haiyun He}
\IEEEauthorblockA{\textit{Internet of Things Thrust, Information Hub}\\
\textit{The Hong Kong University of Science and Technology (Guangzhou)}\\
Guangzhou, China\\
Emails: bli618@connect.hkust-gz.edu.cn \quad  haiyunhe@hkust-gz.edu.cn}
}

\maketitle

\begin{abstract}
Knowledge distillation is widely used to improve generalization in practice, yet its theoretical understanding remains elusive. In the standard distillation setting, a teacher model provides soft predictions to guide the training of a student model. We model teacher  and student  training as coupled stochastic processes and introduce a distillation divergence, defined as the Kullback–Leibler divergence between these two stochastic kernels. 
Within this framework, we derive two generalization bounds for the student model relative to the teacher's generalization gap: an upper bound under a sub-Gaussian assumption via algorithmic stability, and a lower bound under a central condition with sharper dependence on the distillation divergence. We further develop a loss-sharpness-aware bound with an explicit tightness regime, showing that the teacher's local flatness can strictly tighten the bound. Additionally, in a linear Gaussian case study, the distillation divergence admits an interpretable decomposition into bias, variance, and rank-bottleneck costs, yielding practical guidance for distillation design.
\end{abstract}

\begin{IEEEkeywords}
Knowledge Distillation, Generalization Bounds, Algorithm Stability, Sharpness, Stochastic Processes
\end{IEEEkeywords}

\section{Introduction}
Knowledge distillation (KD) \cite{hinton2015distilling} has grown from a model-compression heuristic into a widely adopted approach for improving the generalization performance of neural networks. KD refers to a learning paradigm in which a \emph{teacher} model, typically trained via standard supervised learning, provides auxiliary supervision to guide the training of a \emph{student} model. 
The goal is for the student to leverage information encoded in the teacher’s outputs to achieve improved generalization, even when the student model is more constrained in capacity. Despite extensive empirical success\cite{zhang2017deepmutuallearning,you2017learning}, establishing a rigorous theory that explains when and why distillation helps remains an active line of research. In simplified regimes, including linear and deep linear models, existing analyses show that KD can achieve faster excess-risk decay than standard supervised learning under favorable data geometry and optimization bias, in some cases improving the rate from $O(n^{-1/2})$ to $O(n^{-1})$ \cite{phuong2019towards,ba2014deepnetsreallyneed}. These results connect naturally to Learning Using Privileged Information (LUPI), where the teacher’s soft outputs can be viewed as a proxy for privileged signals that effectively simplify the student’s decision boundary and improve sample efficiency \cite{lopezpaz2016unifyingdistillationprivilegedinformation,JMLR:v16:vapnik15b,hsu2021generalizationboundsdistillation}. Complementary statistical viewpoints model the teacher as an approximation to a Bayes estimator and interpret KD as a form of variance reduction: soft targets smooth gradient estimates relative to hard labels, mitigating stochastic gradient noise and improving the effective bias-variance tradeoff\cite{pmlr-v139-menon21a,safaryan2023knowledgedistillationperformspartial}. More recent work further suggests that, even with imperfect teachers, distillation can behave like learned label smoothing, reducing overconfidence and improving calibration \cite{yuan2021revisitingknowledgedistillationlabel, ji2020knowledgedistillationwideneural}.

Meanwhile, a growing literature relates KD generalization to both loss-landscape geometry and information-theoretic quantities. Empirical and theoretical evidence indicates that distillation can bias optimization toward flatter solutions, often characterized by smaller curvature-related measures such as Hessian trace, which are frequently associated with improved robustness and generalization \cite{pham2022revisitingselfdistillation, 10.1162/neco.1997.9.1.1,zhang2023generalizationmatterslossminima}. This geometric effect is closely aligned with sharpness-aware training principles, including sharpness-aware minimization, which explicitly seeks parameters that are insensitive to small weight perturbations\cite{foret2021sharpnessawareminimizationefficientlyimproving,peng2026leveragingflatnessimproveinformationtheoretic}. In parallel, information-theoretic frameworks analyze KD through representation-level objectives motivated by the information bottleneck principle\cite{tishby2000informationbottleneckmethod,wang2022efficientknowledgedistillationmodel} and variational mutual-information surrogates, aiming to preserve high-entropy “dark knowledge” about inter-class structure\cite{ahn2019variationalinformationdistillationknowledge, tian2022contrastiverepresentationdistillation, ye2024bayesconditionaldistributionestimation}. Despite these advances, existing accounts remain fragmented and still fall short of a unified generalization guarantee that jointly models teacher and student as coupled stochastic training processes, quantifies their mismatch via a distillation divergence, and leverages local sharpness of the attained minima to tighten the bound.

This work quantifies the generalization limits of knowledge distillation by modeling the teacher and student pipelines as two stochastic processes. We introduce a distillation Kullback–Leibler divergence ($\mathsf{K}_n$) and show that it can be decomposed into a dataset-shift term, induced by the teacher-generated pseudo-data, and an algorithm-shift term, induced by the mismatch between teacher and student training kernels. 
Building on this framework, we derive two complementary generalization statements: an upper bound on the student's expected generalization gap under a sub-Gaussian teacher gap, justified via the teacher's algorithm stability, and a lower bound under the ($\eta,c$)-central condition. Together, these inequalities formalize a divergence-deviation tradeoff: when the student process stays close to the teacher process in $\mathsf{K}_n$, the student’s expected generalization gap is constrained to remain teacher-like up to explicit KL-dependent penalties. On the other hand, the lower bound implies that when the teacher generalization gap is large, achieving a substantially smaller student generalization gap requires paying a larger distillation divergence.

Our second contribution incorporates local loss-landscape geometry into the distillation analysis. We derive a sharpness-aware generalization bound whose dominant penalty is controlled by the teacher’s local empirical and population sharpness, quantified via randomized parameter perturbations. We further characterize a tightness radius ($\rho_0$) and identify a regime in which the teacher’s local flatness yields a provably tighter bound than counterparts based on global complexity proxies. Finally, we instantiate the framework in a linear Gaussian example, where $\mathsf{K}_n$ admits an explicit decomposition into teacher bias and variance, a rank-bottleneck approximation cost, and a covariance-mismatch term, yielding a concrete checklist for designing effective distillation protocols.

\section{Preliminaries and Problem Formulation}
\label{sec:formulation}

In this section, we formalize the generalization error of knowledge distillation using an information-theoretic framework. We define the risk functions and model the distillation procedure as a stochastic process governed by conditional distributions.

\subsection{Notation}
The class of Borel probability measures on a set $\calX$ is denoted by $\calP(\calX)$. A random variable $X\sim P_X\in\calP(\calX)$ is called $\sigma^2$-sub-Gaussian, if $\bbE\big[\exp\big(\lambda(X-\bbE[X])\big)\big]\leq \lambda^2\sigma^2/2$ for any $\lambda\in\bbR$. The Kullback-Leibler (KL) divergence between $\mu,\nu\in\calP(\calX)$ ($\mu\ll\nu$) is defined by $\KL(\mu\|\nu)\coloneqq \int   \frac{d  \mu}{d  \nu} \log \frac{d  \mu}{d  \nu} \,  d \nu$. Let $I_d$ denote the $d\times d$ identity matrix.

\subsection{Risk and Generalization Gap Definitions}
\label{subsec:risk_defs}

Consider a data space $\mathcal{Z} := \mathcal{X} \times \mathcal{Y}$, where $\calX$ denotes the feature space and $\calY$ denotes the label set. We assume that the real dataset $D_n = \{Z_i\}_{i=1}^n$ comprises $n$ independently and identically distributed (i.i.d.) samples, where $Z_i=(X_i,Y_i)\sim P\in\calP(\calZ)$ for all $i\in[n]$. Note that the joint distribution of $D_n$ is thus given by $\Delta_{D_n}=P^{\otimes n}$.

We model knowledge distillation as a stochastic process that maps real data to a teacher model and subsequently to a student model. Specifically, we study two learning processes: standard supervised training (the teacher process) and training via knowledge distillation (the student process).
\paragraph{Teacher Model Training Process} Let $\calF\subseteq \{f:\mathcal{X}\to\mathcal{Y}\}$ denote a class of measurable predictors.  We consider a teacher hypothesis class $\calF_T\subset \calF$ consisting of predictive teacher models $f_T:\calX\to\calY$.
A teacher model $f_T\in\calF_T$ is trained on the real dataset $D_n$ via a learning algorithm, which can be characterized by a stochastic kernel $\Delta_{f_T|D_n}\in\calP(\calF_T|\calZ^n)$.
Let $\ell: \mathcal{F} \times \mathcal{Z} \to \mathbb{R}_{\ge 0}$ be a nonnegative loss function. Given any  $f_T\in\calF_T$, the \emph{empirical risk} on the training dataset $D_n$ and the \emph{population risk} with respect to the underlying data distribution $P$ are respectively defined as: 
\begin{equation}
    L_{D_n}(f_T) \coloneqq \frac{1}{n} \sum_{i=1}^n \ell(f_T, Z_i),~
    L_P(f_T) \coloneqq \mathbb{E}_{Z \sim P}[\ell(f_T, Z)].
\end{equation}
To investigate the effects of training process and data distribution on the teacher model generalization capability, we analyze the \emph{expected generalization gap}, defined as the expected gap between the the population and empirical risks:
\begin{equation}
    \gen_T\coloneqq\mathbb{E}_{\Delta_{D_n, f_T}}[L_P(f_T) - L_{D_n}(f_T)]
\end{equation}
where $\Delta_{D_n, f_T} := \Delta_{D_n}\otimes \Delta_{f_T|D_n}$.


\paragraph{Knowledge Distillation via Student Model} 
Consider another student hypothesis class $\calF_S\subset \calF$ consisting of predictive student models $f_S:\calX\to\calY$, which usually has lower complexity than a teacher model.  Given a trained teacher model $f_T$, it synthesizes a pseudo-dataset $\hat{D}_n=\{\hat{Z}_i\}_{i=1}^n$, where $\hat{Z}_i=(X_i,\hatY_i)$ and $\hatY_i=f_T(X_i)$. The joint distribution of $\hatD_n$ is denoted as $\Delta_{\hatD_n}$, which is induced by the joint randomness of $(D_n, f_T)$. 
A student model $f_S\in\calF_S$ is trained on the pseudo-dataset $\hatD_n$ via a learning algorithm characterized by $\Delta_{f_S|\hatD_n}$. We assume that $\Delta_{f_S|\hatD_n}(\cdot|d)\ll \Delta_{f_T|D_n}(\cdot|d)$ for any $d\in\calZ^n$. Similarly, the empirical risk on $\hatD_n$ and population risk with respect to $P$ for a student model $f_S$ are respectively defined as
\begin{equation}
    L_{\hatD_n}(f_S)\coloneqq \frac{1}{n} \sum_{i=1}^n \ell(f_S, \hat{Z}_i),~
    L_P(f_S) \coloneqq \mathbb{E}_{Z \sim P}[\ell(f_S, Z)].
\end{equation}
If the student model successfully distills the knowledge provided by the teacher, its predictions are expected to coincide with the ground truth in expectation. Hence, the generalization capability of the student model is characterized by the expected generalization gap
\begin{equation}
    \gen_S := \mathbb{E}_{\Delta_{\hat{D}_n, f_S}}[L_P(f_S) - L_{\hat{D}_n}(f_S)]
\end{equation}
where $\Delta_{\hatD_n, f_S} := \Delta_{\hatD_n}\otimes \Delta_{f_S|\hatD_n}$.

\section{Generalization Bounds via KL Divergence}
\label{sec:bounds}



We now derive upper and lower generalization bounds for the student model in terms of the teacher model and the KL divergence between the two learning processes, hereafter referred to as the \emph{distillation divergence}. Our analysis relies on the concentration properties of the teacher's generalization gap. We first present a standard upper bound based on sub-Gaussianity and rigorously justify this assumption via algorithmic stability. Subsequently, we explore a lower bound using the central condition. All detailed proofs are provided in Appendix. 

\subsection{Upper Bound under Sub-Gaussianity and  Stability}
\label{subsec:subgaussian}

To relate the student generalization gap $\mathrm{gen}_S$ to  teacher's $\mathrm{gen}_T$ and the distillation divergence, we employ the Donsker-Varadhan (DV) variational representation of the KL divergence.

Let $H(D_n, f_T) := L_P(f_T) - L_{D_n}(f_T)$ denote the generalization gap of the teacher for any data-model pair $(D_n,f_T)$. 
The following theorem establishes the primary bound.

\begin{theorem}[Distillation Generalization Upper Bound]
\label{thm:subgaussian_bound}
Assume that the teacher's generalization gap $H(D_n, f_T)$ is $\sigma^2$-sub-Gaussian under the teacher process $\Delta_{D_n, f_T}$. Then, the student's expected generalization error is bounded by:
\begin{equation}
    \mathrm{gen}_S \le \mathrm{gen}_T + \sigma\sqrt{2\mathsf{K}_n},
    \label{eq:main_bound}
\end{equation}
where $\mathsf{K}_n = \mathrm{KL}(\Delta_{\hat{D}_n, f_S} \| \Delta_{D_n, f_T})$.
\end{theorem}
\begin{proof}[Proof Sketch]
Consider the function $h(d, f) = \lambda(L_P(f) - L_d(f))$ for some $\lambda > 0$. Applying the DV inequality to the change of measure from $\Delta_{\hat{D}_n, f_S}$ to $\Delta_{D_n, f_T}$ yields:
$\lambda \mathrm{gen}_S \le \mathsf{K}_n + \log \mathbb{E}_{\Delta_{D_n, f_T}}\left[ e^{\lambda(L_P(f_T) - L_{D_n}(f_T))} \right]$.

Using the $\sigma^2$-sub-Gaussianity of the centered variable $H - \mathbb{E}[H]$ (where $\mathbb{E}[H] = \mathrm{gen}_T$), we bound the cumulant generating function by $\lambda \mathrm{gen}_T + \frac{\lambda^2\sigma^2}{2}$. Optimizing over $\lambda > 0$ yields the result.
\end{proof}
The KL divergence term $\mathsf{K}_n$ in \eqref{eq:main_bound}  quantifies the mismatch between the student and teacher pipelines. Using the chain rule of KL divergence, we decompose $\mathsf{K}_n$ as follows:
\begin{align}
   & \mathsf{K}_n = \underbrace{\mathrm{KL}(\Delta_{\hat{D}_n} \| \Delta_{D_n})}_{\text{Dataset Shift}} + \mathbb{E}_{\hat{D}_n}  [ \underbrace{\mathrm{KL}(\Delta_{f_S|\hat{D}_n} \| \Delta_{f_T|D_n})}_{\text{Algorithm Shift}} ].
    \label{eq:kl_decomposition}
\end{align}
The first term quantifies the distributional mismatch between real and pseudo data, while the second term captures the algorithmic discrepancy by comparing the student training kernel conditioned on $\hat{D}_n$ with the teacher training kernel conditioned on the real dataset $D_n$, averaged over the joint distribution $\Delta_{\hat{D}_n}$. Intuitively, the order of $\mathsf{K}_n$ with respect to $n$ is dominated by the first term, which may grow linearly in $n$ in certain regimes. In the following discussion and the case study, we show that the bound is non-vacuous under reasonable conditions. Theorem \ref{thm:subgaussian_bound} provides insight into the student distillation process. When $\gen_T$ is small, the teacher model generalizes well and provides an augmented pseudo-dataset $\hatD_n$. Under this case, Theorem \ref{thm:subgaussian_bound} implies that the student distillation process must remain close to the teacher training process in order to achieve better generalization performance.

\textbf{Justification of Sub-Gaussianity:} The sub-Gaussian condition in Theorem \ref{thm:subgaussian_bound} can be satisfied if the teacher learning algorithm is uniformly stable with respect to the training dataset $D_n$.  

\begin{proposition}[Sub-Gaussianity via Stability]
\label{prop:stability}
Assume the loss function is bounded, i.e., $\ell: \calF \times \calZ \to [a, b]$. Suppose the teacher learning algorithm $A_T: \mathcal{Z}^n \to \mathcal{F}$ is $\beta$-uniformly stable, meaning that for any two datasets $D_n, D_n^{(i)}$ differing by a single example $Z_i$,
$\sup_{z \in \mathcal{Z}} |\ell(A_T(D); z) - \ell(A_T(D^{(i)}); z)| \le \beta$.
Then, $H(D_n, f_T)$ is $\sigma^2$-sub-Gaussian with variance proxy:
\begin{equation}
    \sigma^2 = \frac{n}{4}\left(2\beta + \frac{b-a}{n}\right)^2.
\end{equation}
\end{proposition}
\begin{proof}[Proof Sketch] The generalization gap $H(D_n) = L_P(A_T(D_n)) - L_{D_n}(A_T(D_n))$ satisfies the bounded difference property with constant $c_i = 2\beta + \frac{b-a}{n}$. By McDiarmid's inequality, the concentration tail is bounded by $\exp(-2\epsilon^2 / \sum c_i^2)$, which corresponds to the sub-Gaussian variance proxy stated above.\end{proof}
For standard stable algorithms where $\beta$ does not increase with $n$ (typically $\beta = O(1/n)$), we recover $\sigma^2 = O(1/n)$, ensuring the bound in \eqref{eq:main_bound} is non-vacuous.

\subsection{Lower Bounds via the Central Condition}
\label{subsec:central_condition}

In this section, we present a lower bound obtained under another assumption on the tail behavior of the loss, known as the $(\eta, c)$ central condition \cite{xu2017informationtheoreticanalysisgeneralizationcapability,wu2025fastrateinformationtheoreticbounds}. 

\begin{definition}[$(\eta, c)$-Central Condition]
A random variable $H$ satisfies the $(\eta, c)$-central condition under distribution $P$ for $\eta > 0$ and $0 < c \le 1$ if:
\begin{equation}
    \log \mathbb{E}_{P}[e^{-\eta H}] \le -c\eta \mathbb{E}_{P}[H].
\end{equation}
\end{definition}

Applying the DV inequality with the function $-\eta H$ leads to a linear relationship between the student and teacher generalization, distinct from the square-root scaling in \eqref{eq:main_bound}.

\begin{theorem}[Distillation Generalization Lower Bound]\label{thm:lower bound}
If the teacher's generalization gap $H(D_n, f_T)$ satisfies the $(\eta, c)$-central condition under $\Delta_{D_n, f_T}$, then:
\begin{equation}
    \mathrm{gen}_S \ge c \cdot \mathrm{gen}_T - \frac{1}{\eta}\mathsf{K}_n.
    \label{eq:central_bound}
\end{equation}
\end{theorem}
Theorem \ref{thm:lower bound} implies that if the distillation divergence $\mathsf{K}_n$ and $\gen_T$ are both small, the student effectively inherits a significant fraction $c$ of the teacher's generalization performance, coinciding with the implication from Theorem \ref{thm:subgaussian_bound}. 
On the other hand, even when $\mathrm{gen}_T$ is large, it may happen that the distillation pipeline induces a comparably large $\mathsf{K}_n$, which makes the lower bound in \eqref{eq:central_bound} small.
In such cases, the bound leaves room for the student to achieve a small $\mathrm{gen}_S$, consistent with the phenomenon described in the introduction that the student can sometimes generalize better than the teacher.
\begin{remark}[Relationship between Assumptions]
Notably, the central condition is not  entirely decoupled with and is often implied by sub-Gaussianity.
As derived in our analysis, if the random variable $H$ is $\sigma^2$-sub-Gaussian with mean $\bbE[H]=\mathrm{gen}_T > 0$, it satisfies the $(\eta, c)$-central condition provided that $\eta$ is sufficiently small. Specifically, the condition holds if:
\begin{equation}
    c \le 1 - \frac{\eta \sigma^2}{2\mathrm{gen}_T}, \quad \text{and} \quad \eta < \frac{2\mathrm{gen}_T}{\sigma^2}.
\end{equation}
This connection demonstrates that stable teacher algorithms (which are sub-Gaussian) naturally admit this refined characterization for appropriate choices of $\eta$.
\end{remark}

\subsection{A Linear Gaussian Case Study of Distillation Divergence}
\label{subsec:linear_case}

In this section,  we instantiate the distillation divergence $\mathsf{K}_n$ in an a linear Gaussian setting, where real labels are generated by a linear transformation of the features corrupted by Gaussian noise, and the student model is a low-rank approximation of the teacher model. The purpose is to turn the abstract term $\mathsf{K}_n$ into an interpretable checklist of what distillation must control. The emphasis is on an explanatory decomposition rather than on optimizing constants.

We collect the features and labels column-wise into matrices
$X_n \in\mathbb{R}^{d\times n}$ and
$Y_n \in\mathbb{R}^{k\times n}$,
so that the dataset $D_n$ can be equivalently represented by the pair $(X_n,Y_n)$. We assume a noisy label channel $Y_n | X_n \sim \mathcal{MN}(W_\star X_n,\, I_k,\, \nu^2 I_n)$, where $\mathcal{MN}(M,U,V)$ denotes a matrix normal distribution with mean $M$, row covariance $U$, and column covariance $V$, and $\nu$ is a constant noise level.

\paragraph{Teacher as a Gibbs Learner.}
Let the teacher parameter be $W\in\mathbb{R}^{k\times d}$ with Gaussian prior $p_0(W)=\mathcal{MN}(0,I_k,\lambda^{-1}I_d)$. Given $D_n$, define the Gibbs posterior with inverse temperature $\beta_T$ by
$q_T(W\mid D_n)\propto p_0(W)\exp\!\big(-\frac{\beta_T}{2\nu^2}\|Y_n-WX_n\|_F^2\big)$.
Completing the square yields $q_T(W\mid D_n)=\mathcal{MN}(\bar W_T, I_k, \Sigma_T)$ with
$\Sigma_T=\big(\lambda I_d+\frac{\beta_T}{\nu^2}X_nX_n^\top\big)^{-1}$ and $\bar W_T=\frac{\beta_T}{\nu^2}Y_nX_n^\top\Sigma_T$.

\paragraph{Low-Rank Student and a Rank Bottleneck.}
Sample $W_T\sim q_T(\cdot\mid D_n)$ and generate pseudo labels via the same noisy channel
$\hat Y_n\mid (W_T,X_n)\sim \mathcal{MN}(W_T X_n, I_k, \nu^2 I_n)$, given $\hat D_n=(X_n,\hat Y_n)$.
To model an explicit student capacity constraint, introduce a rank-$\kappa$ map $M^\star(W_T,X_n)$ as a best rank-$\kappa$ approximation in prediction space:
$M^\star(W_T,X_n)\in \arg\min_{M:\,\mathrm{rank}(M)=\kappa}\|W_T X_n - W_T M X_n\|_F^2$.
Using the SVD-based solution from \cite{chen2021drone}, we define the student as a local Gaussian centered at the projected parameters $\Theta^\circ = W_T M^\star$, specifically $q_S(\Theta\mid W_T,X_n)=\mathcal{MN}(\Theta^\circ, I_k, \Sigma_S)$ for some $\Sigma_S\succ 0$.

First, for the dataset-shift term, convexity of the KL divergence allows us to bound the mismatch using the teacher's prediction error. This error decomposes into the teacher's bias and variance on the observed design: \begin{align}
&\KL(\Delta_{\hat D_n}\,\|\,\Delta_{D_n})
\le \frac{1}{2\nu^2}\,\E_{D_n}\E_{W_T\mid D_n}\big[\|(W_T-W_\star)X_n\|_F^2\big] \nonumber \\
&= \E_{D_n}\left[ \frac{1}{2\nu^2}\big(\mathrm{Bias}(D_n)+k\,\mathrm{Var}(D_n)\big) \right],
\label{eq:lin_data_term}
\end{align}
where $\mathrm{Bias}(D_n)\coloneqq\|(\bar W_T-W_\star)X_n\|_F^2$ and $\mathrm{Var}(D_n)\coloneqq\mathrm{tr}(X_n^\top\Sigma_T X_n)$.

Second, we bound the algorithm-shift by analyzing the student's residual in prediction space and the geometry of the posteriors:
\begin{align}
&\quad \E\Big[\KL(\Delta_{f_S\mid \hat D_n}\,\|\,\Delta_{f_T\mid D_n})\Big]
\le  \nonumber \\ 
&\E_{D_n}\Big[\mathrm{Apx}(D_n)+\mathrm{Cov}(\Sigma_S,\Sigma_T)+\mathrm{Spread}(D_n)\Big].
\label{eq:lin_alg_term}
\end{align}
Here, the terms are defined as follows: 
$\mathrm{Apx}(D_n) \coloneqq \lambda\,\mathbb E_{W_T\mid D_n}\|W_T(I-M^\star)\|_F^2+\frac{\beta_T}{\nu^2}\,\mathbb E_{W_T\mid D_n}\|W_T(I-M^\star)X_n\|_F^2$ quantifies the energy discarded by the rank-$\kappa$ bottleneck. 
$\mathrm{Cov}(\Sigma_S,\Sigma_T) \coloneqq \frac{k}{2}\big(\mathrm{tr}(\Sigma_T^{-1}\Sigma_S)-d+\log\frac{\det\Sigma_T}{\det\Sigma_S}\big)$ measures the misalignment between student and teacher covariance structures. 
$\mathrm{Spread}(D_n) \coloneqq k\big(\lambda\,\mathrm{tr}(\Sigma_T)+\frac{\beta_T}{\nu^2}\mathrm{tr}(X_n^\top\Sigma_T X_n)\big)$ represents the residual contribution from teacher posterior fluctuations.

Combining \eqref{eq:lin_data_term} and \eqref{eq:lin_alg_term} yields the compact bound:
\begin{align}
\mathsf{K}_n
\ &\le\
\E_{D_n}\!\Big[
\underbrace{\frac{1}{2\nu^2}\big(\mathrm{Bias}(D_n)+k\,\mathrm{Var}(D_n)\big)}_{\text{Teacher Prediction Error}}
+ \nonumber\\
&\underbrace{\mathrm{Apx}(D_n)}_{\text{Student Capacity}}
+\underbrace{\mathrm{Cov}(\Sigma_S,\Sigma_T)}_{\text{Geometry Mismatch}}
+\underbrace{\mathrm{Spread}(D_n)}_{\text{Posterior Spread}}
\Big].
\label{eq:lin_total_KL_compact}
\end{align}
This decomposition makes the trade-offs explicit: reducing teacher bias/variance improves the data term, while increasing the effective rank $\kappa$ minimizes $\mathrm{Apx}(D_n)$. Furthermore, aligning the student's sampling geometry $\Sigma_S$ with $\Sigma_T$ minimizes the covariance penalty.

Finally, we examine how these terms scale with sample size $n$. The approximation term $\mathrm{Apx}(D_n)$, viewed as a Frobenius norm over a $k \times n$ prediction matrix, grows linearly with $n$ due to accumulating residuals. This is acceptable since the final bound includes a normalization factor like $\sqrt{\mathsf{K}_n/n}$, yielding an overall $O(1)$ error. In contrast, the teacher-dependent terms ($\mathrm{Bias}$ and $\mathrm{Var}$) scale more favorably: posterior contraction of $W_T$ offsets the linear growth, keeping their contribution at $O(1)$. This result verifies that the bound in Theorem \ref{thm:subgaussian_bound} is non-vacuous. 

\subsection{Sharpness-Aware Bounds and a Tightness Radius}
\label{subsec:sharpness}

In this section, we aim to introduce the effect of local loss-landscape geometry into the distillation analysis and obtain a tighter generalization bound.
We quantify local flatness via randomized perturbations. Let $U$ be uniformly
distributed on the ball $\{u:\|u\|\le\rho\}$ and independent of all other
randomness. For any dataset $D\in\mathcal Z^n$ and model $f$, define the
empirical and population sharpness
\begin{align}
S_D(f)
&:= \E_U\!\big[L_D(f+U)\big]-L_D(f),
\label{eq:sharp_emp}\\
S_P(f)
&:= \E_U\!\big[L_P(f+U)\big]-L_P(f).
\label{eq:sharp_pop}
\end{align}
These quantities measure the average increase of risk under radius-$\rho$
perturbations, so smaller values indicate a locally flatter landscape at the
chosen scale.

To justify concentration without assuming sub-Gaussianity \emph{a priori}, we
impose mild boundedness, smoothness, and stability conditions. We assume the
loss $\ell$ is bounded in $[a,b]$ (A1) and globally $L$-Lipschitz (A3). Around
the teacher output $f_T=A_T(D_n)$, we assume local regularity on
$B(f_T,2\rho)$: $\ell(\cdot;z)$ is $\alpha$-smooth and its gradients are
uniformly bounded by $g_0$ (A4), yielding a refined local Lipschitz scale
$g_0+\alpha\rho$ at radius $\rho$. We also assume the teacher algorithm
satisfies parameter stability (A2), namely for any neighboring datasets
$D,D^{(i)}$ differing in one example,
$\|A_T(D)-A_T(D^{(i)})\|\le \kappa/n$. Finally, to control population
sharpness, we assume bounded local curvature of the population risk on
$B(f_T,\rho)$ (A5), i.e.,
$\sup_{\|v\|\le\rho}\|\nabla^2 L_P(f_T+v)\|_{\op}\le \tau_{\op}$, which implies
a quadratic upper bound on $\E[S_P(f_T)]$ at scale $\rho$. These conditions,
together with bounded differences and McDiarmid's inequality, yield
sub-Gaussian concentration for the perturbed teacher generalization gap and
for the empirical sharpness, with variance proxies that depend explicitly on
$(\rho,n)$ and the constants $(L,g_0,\alpha,\kappa,\tau_{\op})$.

Our goal is to upper bound the student generalization $\gen_S$ in terms of
teacher quality, teacher local flatness, and the distillation divergence
$\mathsf{K}_n=\KL(\Delta_{\hat D_n,f_S}\|\Delta_{D_n,f_T})$. A key point is
that $\mathsf{K}_n$ is the only KL quantity used here, so it captures, in a
single term, both the pseudo-data generation $\Delta_{\hatD_n\mid D_n}$ and the
student learning process $\Delta_{f_S\mid \hat D_n}$.

We adopt the standard flatness-based simplification that the returned student
is not improved, on average, by a small random perturbation at radius $\rho$:
$\E_{\hat D_n,f_S}[L_P(f_S)]\le \E_{\hat D_n,f_S,U}[L_P(f_S+U)]$. This avoids an
almost-sure condition on the randomized output while providing the inequality
needed to compare unperturbed and perturbed risks in expectation.


Let $h(D,f):=L_P(f)-L_D(f)$ and define the perturbed gap
$h_U(D,f):=\E_U[h(D,f+U)]$. The next theorem summarizes the resulting guarantee.

\begin{theorem}[Sharpness-Aware Distillation Generalization Bound]
\label{thm:sharp_main}
Suppose that under the teacher process $\Delta_{D_n,f_T}$, both $h_U(D_n,f_T)$ and $S_{D_n}(f_T)$ admit sub-Gaussian moment bounds
with proxies $\sigma_u^2$ and $\nu^2$, respectively. Then
\begin{equation}
\gen_S
\le
\gen_T
+\E_{D_n,f_T}\!\big[S_P(f_T)\big]
+(\sigma_u+\nu)\sqrt{2\mathsf{K}_n}.
\label{eq:sharp_bound_main}
\end{equation}
\end{theorem}
\begin{proof}[Proof Sketch] By the local optimality condition and adding/subtracting the perturbed
empirical term, we obtain the decomposition
$\gen_S \le \E_{\hat D_n,f_S}[H_\rho(\hat D_n,f_S)] + \E_{\hat D_n,f_S}[S_{\hat D_n}(f_S)]$,
where $H_\rho(D,f)=\E_U[L_P(f+U)-L_D(f+U)]$. A change-of-measure inequality
(Donsker--Varadhan) transfers both expectations to the teacher process,
paying $\sigma_u\sqrt{2\mathsf{K}_n}$ and $\nu\sqrt{2\mathsf{K}_n}$,
respectively. Finally,
$\E_{D_n,f_T}[H_\rho(D_n,f_T)]=\gen_T+\E[S_P(f_T)]-\E[S_{D_n}(f_T)]$,
so the $-\E[S_{D_n}(f_T)]$ term cancels with the teacher-side control of
$\E[S_{\hat D_n}(f_S)]$, yielding \eqref{eq:sharp_bound_main}. \end{proof} Eq. \eqref{eq:sharp_bound_main} decomposes the error into three components: direct transfer of teacher quality ($\gen_T$), an explicit local-geometry term ($\E[S_P(f_T)]$), and a process-mismatch penalty $\sqrt{\mathsf{K}_n}$ scaled by $(\sigma_u+\nu)$. Since $\sigma_u(\rho)$ substitutes the global Lipschitz scale with a local measure around $f_T$, it may be strictly smaller than the counterpart found in standard global bounds. This motivates a direct comparison against the baseline, revealing a radius regime in which accounting for local geometry provides a provably tighter guarantee.

We compare \eqref{eq:sharp_bound_main} to a baseline bound that depends only
on global geometry. Define the standard bound
$B_{\mathrm{std}}:=\gen_T+\sigma_0\sqrt{2\mathsf{K}_n}$ and the
sharpness-aware bound
$B_{\mathrm{sh}}(\rho):=\gen_T+\E[S_P(f_T)]+\big(\sigma_u(\rho)+\nu(\rho)\big)\sqrt{2\mathsf{K}_n}$.

Under the standing stability and smoothness assumptions, one obtains the
convenient proxies
\begin{align}
\sigma_0
&= \frac{1}{\sqrt{n}}\Big(\kappa L+\frac{b-a}{2}\Big),
\label{eq:sigma0_short}\\
\sigma_u(\rho)
&= \frac{1}{\sqrt{n}}\Big(\kappa(g_0+\alpha\rho)+\frac{b-a}{2}\Big),
\label{eq:sigmau_short}\\
\nu(\rho)
&= \frac{1}{\sqrt{n}}\Big(
\frac{\alpha\kappa}{2}\rho + g_0\rho + \frac{\alpha}{2}\rho^2
\Big),
\label{eq:nu_short}\\
\E[S_P(f_T)]
&\le \frac{1}{2}\tau_{\op}\rho^2.
\label{eq:SP_short}
\end{align}
The structure is intuitive: $\sigma_u(\rho)$ improves upon $\sigma_0$ by
replacing the global Lipschitz scale $L$ with the local scale
$g_0+\alpha\rho$, while $\nu(\rho)$ and $\E[S_P(f_T)]$ grow with $\rho$,
reflecting the increasing cost of probing larger neighborhoods.

\textbf{A sufficient condition for strict improvement:}
Let $\Delta:=L-g_0$ denote the local flatness gap. When $\Delta>0$, the
teacher neighborhood around $f_T$ is flatter than the global worst case.
Define
\begin{equation}
A_0:=\kappa\Delta,\quad
A_1:=\frac{3}{2}\alpha\kappa+g_0,\quad
A_2:=\frac{\alpha}{2}.
\label{eq:A012}
\end{equation}
A short calculation from \eqref{eq:sigma0_short}--\eqref{eq:nu_short} gives
\begin{equation}
\sigma_0-\sigma_u(\rho)-\nu(\rho)
=
\frac{1}{\sqrt{n}}\big(A_0-A_1\rho-A_2\rho^2\big).
\label{eq:sigmadiff}
\end{equation}
Using \eqref{eq:SP_short}, a sufficient condition for
$B_{\mathrm{sh}}(\rho)<B_{\mathrm{std}}$ is
\begin{equation}
\frac{1}{\sqrt{n}}\big(A_0-A_1\rho-A_2\rho^2\big)\sqrt{2\mathsf{K}_n}
>
\frac{1}{2}\tau_{\op}\rho^2.
\label{eq:suff_cond}
\end{equation}

\begin{corollary}[A Tightening Interval]
\label{cor:interval}
Assume $\Delta>0$. Define
\begin{align}
c_0&:=A_0\sqrt{2\mathsf{K}_n}, \label{eq:c0}\\
c_1&:=A_1\sqrt{2\mathsf{K}_n}, \label{eq:c1}\\
c_2&:=A_2\sqrt{2\mathsf{K}_n}+\frac{1}{2}\tau_{\op}\sqrt{n}. \label{eq:c2}
\end{align}
Let $\rho_0$ be the unique positive root of $c_2\rho^2+c_1\rho-c_0=0$: 
\begin{equation}
\rho_0
=
\frac{-c_1+\sqrt{c_1^2+4c_2c_0}}{2c_2}.
\label{eq:rho0}
\end{equation}
Then for all $0<\rho<\rho_0$, \eqref{eq:suff_cond} holds and hence
$B_{\mathrm{sh}}(\rho)<B_{\mathrm{std}}$.
\end{corollary}

The left side of \eqref{eq:suff_cond} decreases with $\rho$ due to the
negative $A_1\rho$ and $A_2\rho^2$ terms, while the right side grows as
$\rho^2$. When $\Delta>0$, the constant gain term $A_0$ is positive, so the
inequality holds at sufficiently small radii. As $\rho$ increases, quadratic
growth in sharpness eventually dominates, producing a unique cutoff radius
$\rho_0$. The interval $(0,\rho_0)$ formalizes the idea that local geometry
helps only within a neighborhood where the teacher is sufficiently flat.

Eq.~\eqref{eq:sharp_bound_main} suggests three aligned design principles:
improving teacher flatness at small scales reduces $\E[S_P(f_T)]$, tightening
the coupling between student and teacher processes reduces $\mathsf{K}_n$,
and the perturbation radius $\rho$ should be treated as a local scale
parameter. Corollary~\ref{cor:interval} makes the last point explicit by
exhibiting a provable range of radii that guarantees a tighter bound whenever
$\Delta>0$.

\section{Conclusion and Future work}
In this work, we established a unified information-theoretic framework for knowledge distillation by introducing the distillation divergence ($\mathsf{K}_n$) to quantify mismatches in teacher training and student distillation processes. We derived stability-based upper bounds and central-condition lower bounds, demonstrating that a teacher's local flatness can strictly tighten generalization guarantees, while a linear Gaussian case study provided interpretable decompositions of $\mathsf{K}_n$ into bias, variance, and rank-bottleneck terms. Future work will focus on three key directions: conducting experimental validation on large-scale benchmarks to verify the correlation between $\mathsf{K}_n$ and generalization gaps, refining the theory to close the gap between the upper and lower bounds, and designing new algorithms that explicitly minimize the derived divergence components to improve student performance.

\bibliographystyle{IEEEtran}
\bibliography{ISIT_references}

\onecolumn
\section{Appendix}
\subsection{Preliminary Tools}

\subsubsection{Donsker--Varadhan change of measure inequality}
\begin{lemma}[Donsker--Varadhan inequality]
\label{lem:dv}
Let $P,Q$ be probability measures on the same measurable space with $Q\ll P$.
For any measurable function $g$ with $\bbE_P[e^{g}]<\infty$,
\begin{equation}
\bbE_Q[g]
\le
\KL(Q\|P) + \log \bbE_P[e^{g}].
\label{eq:dv}
\end{equation}
\end{lemma}

\begin{proof}
Define the Radon--Nikodym derivative $r := \frac{dQ}{dP}$.
Then $\bbE_P[r]=1$ and
\[
\KL(Q\|P)=\bbE_Q\!\left[\log \frac{dQ}{dP}\right]=\bbE_P[r\log r].
\]
By Jensen's inequality applied to the convex function $\log$,
\[
\log \bbE_P[e^{g}]
=
\log \bbE_P\!\left[r\cdot \frac{e^{g}}{r}\right]
\ge
\bbE_P\!\left[r \log \frac{e^{g}}{r}\right]
=
\bbE_Q[g] - \bbE_P[r\log r].
\]
Rearranging gives \eqref{eq:dv}.
\end{proof}

\subsubsection{A standard sub-Gaussian mgf bound}
\begin{lemma}[Sub-Gaussian mgf bound]
\label{lem:subg_mgf}
Let $X$ be $\sigma^2$-sub-Gaussian, meaning
$\log \bbE[\exp(\lambda(X-\bbE[X]))] \le \tfrac{\lambda^2\sigma^2}{2}$ for all $\lambda\in\mathbb{R}$.
Then for any $\lambda\in\mathbb{R}$,
\begin{equation}
\log \bbE[e^{\lambda X}]
\le
\lambda \bbE[X] + \frac{\lambda^2\sigma^2}{2}.
\label{eq:subg_mgf}
\end{equation}
\end{lemma}

\begin{proof}
By definition,
\[
\bbE[e^{\lambda X}]
=
\bbE\!\left[e^{\lambda (X-\bbE[X])}\right]\cdot e^{\lambda \bbE[X]}.
\]
Taking logs and applying the sub-Gaussian condition gives \eqref{eq:subg_mgf}.
\end{proof}

\subsection{Proof of Theorem 1}

\subsubsection{Statement}
\begin{theorem}[Distillation generalization upper bound]
\label{thm:upper}
Assume that $h(D_n,f_T)$ is $\sigma^2$-sub-Gaussian under $\Delta_{D_n,f_T}$.
Then
\begin{equation}
\mathrm{gen}_S \le \mathrm{gen}_T + \sigma\sqrt{2\mathsf{K}_n}.
\label{eq:upper_goal}
\end{equation}
\end{theorem}

\subsubsection{Proof}
\begin{proof}
Fix any $\lambda>0$ and choose in Lemma~\ref{lem:dv}
\[
P = \Delta_{D_n,f_T}, \qquad Q = \Delta_{\hat{D}_n,f_S}, \qquad g(d,f)=\lambda\, h(d,f).
\]
Then \eqref{eq:dv} gives
\begin{equation}
\bbE_{\Delta_{\hat{D}_n,f_S}}[\lambda h(\hat{D}_n,f_S)]
\le
\KL(\Delta_{\hat{D}_n,f_S}\|\Delta_{D_n,f_T})
+
\log \bbE_{\Delta_{D_n,f_T}}[e^{\lambda h(D_n,f_T)}].
\label{eq:dv_applied}
\end{equation}
By definitions, the left side equals $\lambda \mathrm{gen}_S$ and the KL term equals $\mathsf{K}_n$, so
\begin{equation}
\lambda \mathrm{gen}_S
\le
\mathsf{K}_n
+
\log \bbE_{\Delta_{D_n,f_T}}[e^{\lambda h(D_n,f_T)}].
\label{eq:dv_applied_simplified}
\end{equation}

Since $h(D_n,f_T)$ is $\sigma^2$-sub-Gaussian with mean $\mathrm{gen}_T$,
Lemma~\ref{lem:subg_mgf} yields
\begin{equation}
\log \bbE_{\Delta_{D_n,f_T}}[e^{\lambda h(D_n,f_T)}]
\le
\lambda \mathrm{gen}_T + \frac{\lambda^2\sigma^2}{2}.
\label{eq:mgf_bound}
\end{equation}

Substitute \eqref{eq:mgf_bound} into \eqref{eq:dv_applied_simplified}:
\[
\lambda \mathrm{gen}_S
\le
\mathsf{K}_n + \lambda \mathrm{gen}_T + \frac{\lambda^2\sigma^2}{2}.
\]
Rearrange:
\begin{equation}
\mathrm{gen}_S - \mathrm{gen}_T
\le
\frac{\mathsf{K}_n}{\lambda} + \frac{\lambda\sigma^2}{2}.
\label{eq:opt_lambda_obj}
\end{equation}

Consider the function $\phi(\lambda)=\frac{\mathsf{K}_n}{\lambda}+\frac{\lambda\sigma^2}{2}$ for $\lambda>0$.
Differentiate:
\[
\phi'(\lambda) = -\frac{\mathsf{K}_n}{\lambda^2} + \frac{\sigma^2}{2}.
\]
Set $\phi'(\lambda)=0$:
\[
-\frac{\mathsf{K}_n}{\lambda^2} + \frac{\sigma^2}{2}=0
\quad\Longrightarrow\quad
\lambda^2 = \frac{2\mathsf{K}_n}{\sigma^2}
\quad\Longrightarrow\quad
\lambda^\star = \frac{\sqrt{2\mathsf{K}_n}}{\sigma}.
\]
Plugging $\lambda^\star$ into \eqref{eq:opt_lambda_obj} gives
\[
\mathrm{gen}_S - \mathrm{gen}_T
\le
\frac{\mathsf{K}_n}{\sqrt{2\mathsf{K}_n}/\sigma} + \frac{(\sqrt{2\mathsf{K}_n}/\sigma)\sigma^2}{2}
=
\sigma\sqrt{\frac{\mathsf{K}_n}{2}} + \sigma\sqrt{\frac{\mathsf{K}_n}{2}}
=
\sigma\sqrt{2\mathsf{K}_n}.
\]
This is exactly \eqref{eq:upper_goal}.
\end{proof}


\subsection{Proof of Proposition 1}

\subsubsection{Statement}
\begin{proposition}[Sub-Gaussianity via stability]
\label{prop:stability}
Assume the loss is bounded: $\ell(f,z)\in[a,b]$ for all $f,z$.
Assume the teacher algorithm $A_T$ is $\beta$-uniformly stable:
for any neighboring datasets $D_n$ and $D_n^{(i)}$ differing in one example,
\[
\sup_{z\in\calZ}\big|\ell(A_T(D_n),z)-\ell(A_T(D_n^{(i)}),z)\big|\le \beta.
\]
Let
\[
H(D_n):=L_P(A_T(D_n)) - L_{D_n}(A_T(D_n)).
\]
Then $H(D_n)$ satisfies bounded differences with constants
\begin{equation}
c_i = 2\beta + \frac{b-a}{n}\quad \text{for all } i\in[n],
\label{eq:ci}
\end{equation}
and hence is sub-Gaussian with variance proxy
\begin{equation}
\sigma^2=\frac{1}{4}\sum_{i=1}^n c_i^2
=
\frac{n}{4}\left(2\beta+\frac{b-a}{n}\right)^2.
\label{eq:sigma_stability}
\end{equation}
\end{proposition}

\subsubsection{Proof}
\begin{proof}
Let $D_n=(z_1,\dots,z_n)$ and let $D_n^{(i)}=(z_1,\dots,z_{i-1},z_i',z_{i+1},\dots,z_n)$
be a neighboring dataset. Define
\[
f := A_T(D_n),\qquad f':=A_T(D_n^{(i)}).
\]
Then
\[
H(D_n)=\bbE_{Z\sim P}[\ell(f,Z)] - \frac{1}{n}\sum_{j=1}^n \ell(f,z_j),
\]
\[
H(D_n^{(i)})=\bbE_{Z\sim P}[\ell(f',Z)] - \frac{1}{n}\sum_{j=1}^n \ell(f',z_j^{(i)}),
\]
where $z_j^{(i)}=z_j$ for $j\neq i$ and $z_i^{(i)}=z_i'$.

By uniform stability and taking expectation over $Z\sim P$,
\begin{equation}
\Big|\bbE_{Z\sim P}[\ell(f,Z)] - \bbE_{Z\sim P}[\ell(f',Z)]\Big|
\le
\bbE_{Z\sim P}\big[|\ell(f,Z)-\ell(f',Z)|\big]
\le
\beta.
\label{eq:pop_diff}
\end{equation}

Consider
\[
\left|\frac{1}{n}\sum_{j=1}^n \ell(f,z_j) - \frac{1}{n}\sum_{j=1}^n \ell(f',z_j^{(i)})\right|
\le
\frac{1}{n}\sum_{j=1}^n \big|\ell(f,z_j)-\ell(f',z_j^{(i)})\big|.
\]
Now split the sum into $j\neq i$ and $j=i$.

For $j\neq i$, we have $z_j^{(i)}=z_j$, so by stability
\begin{equation}
|\ell(f,z_j)-\ell(f',z_j)|\le \beta.
\label{eq:emp_diff_jneqi}
\end{equation}

For $j=i$, we use triangle inequality:
\[
|\ell(f,z_i)-\ell(f',z_i')|
\le
|\ell(f,z_i)-\ell(f',z_i)| + |\ell(f',z_i)-\ell(f',z_i')|.
\]
The first term is bounded by $\beta$ by stability, and the second term is bounded by $(b-a)$ by loss boundedness.
Thus
\begin{equation}
|\ell(f,z_i)-\ell(f',z_i')|\le \beta + (b-a).
\label{eq:emp_diff_i}
\end{equation}

Combine \eqref{eq:emp_diff_jneqi} and \eqref{eq:emp_diff_i}:
\[
\left|\frac{1}{n}\sum_{j=1}^n \ell(f,z_j) - \frac{1}{n}\sum_{j=1}^n \ell(f',z_j^{(i)})\right|
\le
\frac{(n-1)\beta + (\beta+(b-a))}{n}
=
\beta+\frac{b-a}{n}.
\]
So
\begin{equation}
\Big|L_{D_n}(f) - L_{D_n^{(i)}}(f')\Big|
\le
\beta+\frac{b-a}{n}.
\label{eq:emp_diff}
\end{equation}

Using \eqref{eq:pop_diff} and \eqref{eq:emp_diff},
\[
|H(D_n)-H(D_n^{(i)})|
\le
\underbrace{\big|L_P(f)-L_P(f')\big|}_{\le \beta}
+
\underbrace{\big|L_{D_n}(f)-L_{D_n^{(i)}}(f')\big|}_{\le \beta+(b-a)/n}
\le
2\beta+\frac{b-a}{n}.
\]
This establishes the bounded-differences constants \eqref{eq:ci}.

By McDiarmid's inequality, for all $t>0$,
\[
\bbP\big(H(D_n)-\bbE[H(D_n)]\ge t\big)
\le
\exp\!\left(-\frac{2t^2}{\sum_{i=1}^n c_i^2}\right).
\]
A standard equivalence between McDiarmid tails and sub-Gaussian mgf yields that
$H(D_n)$ is sub-Gaussian with variance proxy $\sigma^2=\frac{1}{4}\sum_{i=1}^n c_i^2$.
With constant $c_i$ from \eqref{eq:ci}, this gives \eqref{eq:sigma_stability}.
\end{proof}

\subsection{Proof of Theorem 2}

\subsubsection{Central condition definition}
\begin{definition}[$(\eta,c)$-central condition]
A random variable $X$ satisfies the $(\eta,c)$-central condition under $P$
if $\eta>0$ and $0<c\le 1$ and
\begin{equation}
\log \bbE_{P}[e^{-\eta X}] \le -c\eta\, \bbE_{P}[X].
\label{eq:central_def}
\end{equation}
\end{definition}

\subsubsection{Statement}
\begin{theorem}[Distillation generalization lower bound]
\label{thm:lower}
Assume $h(D_n,f_T)$ satisfies the $(\eta,c)$-central condition under $\Delta_{D_n,f_T}$.
Then
\begin{equation}
\mathrm{gen}_S \ge c\cdot \mathrm{gen}_T - \frac{1}{\eta}\mathsf{K}_n.
\label{eq:lower_goal}
\end{equation}
\end{theorem}

\subsubsection{Proof}
\begin{proof}
Apply Lemma~\ref{lem:dv} with
\[
P=\Delta_{D_n,f_T},\quad Q=\Delta_{\hat{D}_n,f_S},\quad g(d,f)=-\eta\, h(d,f).
\]
Then
\begin{equation}
\bbE_{\Delta_{\hat{D}_n,f_S}}[-\eta\, h(\hat{D}_n,f_S)]
\le
\mathsf{K}_n
+
\log \bbE_{\Delta_{D_n,f_T}}[e^{-\eta\, h(D_n,f_T)}].
\label{eq:dv_negative}
\end{equation}

By assumption, $h(D_n,f_T)$ satisfies \eqref{eq:central_def} with $X=h(D_n,f_T)$ and $P=\Delta_{D_n,f_T}$:
\begin{equation}
\log \bbE_{\Delta_{D_n,f_T}}[e^{-\eta\, h(D_n,f_T)}]
\le
-c\eta\, \bbE_{\Delta_{D_n,f_T}}[h(D_n,f_T)]
=
-c\eta\, \mathrm{gen}_T.
\label{eq:central_used}
\end{equation}

Substitute \eqref{eq:central_used} into \eqref{eq:dv_negative}:
\[
-\eta\, \bbE_{\Delta_{\hat{D}_n,f_S}}[h(\hat{D}_n,f_S)]
\le
\mathsf{K}_n - c\eta\, \mathrm{gen}_T.
\]
Divide by $-\eta$ (note $\eta>0$ flips the inequality):
\[
\bbE_{\Delta_{\hat{D}_n,f_S}}[h(\hat{D}_n,f_S)]
\ge
c\,\mathrm{gen}_T - \frac{1}{\eta}\mathsf{K}_n.
\]
Since the left side is exactly $\mathrm{gen}_S$, we obtain \eqref{eq:lower_goal}.
\end{proof}

\subsubsection{Sub-Gaussianity implies a central condition for small $\eta$}
\label{subsubsec:subGau2central}
\begin{remark}[Deriving a valid $(\eta,c)$ from sub-Gaussianity]
\label{rem:subg_to_central}
Assume $X$ is $\sigma^2$-sub-Gaussian with mean $\mu=\bbE[X]>0$.
Then for any $\eta>0$,
\[
\log \bbE[e^{-\eta X}]
\le
-\eta \mu + \frac{\eta^2\sigma^2}{2}
=
-\eta\mu\left(1-\frac{\eta\sigma^2}{2\mu}\right).
\]
Thus $X$ satisfies the $(\eta,c)$-central condition with
\[
c \le 1-\frac{\eta\sigma^2}{2\mu}
\quad \text{provided that} \quad
0<\eta<\frac{2\mu}{\sigma^2}.
\]
\end{remark}


\subsection{Linear Gaussian Case Study (Detailed KL Decomposition)}

\subsubsection{Matrix normal definition and basic identities}
\begin{definition}[Matrix normal distribution]
A random matrix $A\in\mathbb{R}^{k\times n}$ follows a matrix normal distribution
$A\sim \mathcal{MN}(M, U, V)$ if
\begin{equation}
\mathrm{vec}(A) \sim \mathcal{N}(\mathrm{vec}(M),\, V\otimes U),
\label{eq:mn_def}
\end{equation}
where $M\in\mathbb{R}^{k\times n}$, $U\in\mathbb{R}^{k\times k}$, $V\in\mathbb{R}^{n\times n}$.
\end{definition}

We use the vectorization identity
\begin{equation}
\mathrm{vec}(W X) = (X^\top \otimes I_k)\,\mathrm{vec}(W),
\label{eq:vec_WX}
\end{equation}
valid for $W\in\mathbb{R}^{k\times d}$ and $X\in\mathbb{R}^{d\times n}$.

For Gaussians with the same covariance, we use
\begin{equation}
\KL(\mathcal{N}(m_1,\Sigma)\,\|\,\mathcal{N}(m_0,\Sigma))
=\frac{1}{2}(m_1-m_0)^\top \Sigma^{-1}(m_1-m_0).
\label{eq:gauss_samecov_kl}
\end{equation}

\subsubsection{Generative model}
Collect features and labels column-wise into
$X\in\mathbb{R}^{d\times n}$ and $Y\in\mathbb{R}^{k\times n}$ so that $D_n=(X,Y)$.
Assume a noisy linear label channel
\begin{equation}
Y\mid X \sim \mathcal{MN}(W_\star X,\ I_k,\ \nu^2 I_n),
\label{eq:lin_gen}
\end{equation}
where $W_\star\in\mathbb{R}^{k\times d}$ is the ground-truth linear map and $\nu>0$ is the noise level.

\subsubsection{Teacher as a Gibbs learner (closed form)}
Let the teacher parameter be $W\in\mathbb{R}^{k\times d}$ with prior
\begin{equation}
p_0(W)=\mathcal{MN}(0,\ I_k,\ \lambda^{-1}I_d).
\label{eq:prior}
\end{equation}
Given $D_n=(X,Y)$, define the Gibbs posterior with inverse temperature $\beta_T$:
\begin{equation}
q_T(W\mid D_n)\propto p_0(W)\exp\!\left(-\frac{\beta_T}{2\nu^2}\|Y-WX\|_F^2\right).
\label{eq:gibbs_teacher}
\end{equation}

\begin{lemma}[Closed form of $q_T$]
\label{lem:qT_closed}
The posterior is matrix normal:
\begin{equation}
q_T(W\mid D_n)=\mathcal{MN}(\bar W_T,\ I_k,\ \Sigma_T),
\qquad
\Sigma_T=\Big(\lambda I_d+\frac{\beta_T}{\nu^2}XX^\top\Big)^{-1},
\qquad
\bar W_T=\frac{\beta_T}{\nu^2}YX^\top \Sigma_T.
\label{eq:qT_closed}
\end{equation}
\end{lemma}

\begin{proof}
Let $w=\mathrm{vec}(W)$, $y=\mathrm{vec}(Y)$, and $B_X:=X^\top\otimes I_k$.
Then $\mathrm{vec}(WX)=B_X w$ and $\|Y-WX\|_F^2=\|y-B_X w\|_2^2$.
The prior implies $w\sim\mathcal{N}(0,\lambda^{-1}I_{kd})$.
Hence the (unnormalized) log density is quadratic in $w$ with precision
\[
\lambda I_{kd}+\frac{\beta_T}{\nu^2}B_X^\top B_X
=
\lambda I_{kd}+\frac{\beta_T}{\nu^2}(XX^\top\otimes I_k),
\]
so the covariance is $(\lambda I_d+\frac{\beta_T}{\nu^2}XX^\top)^{-1}\otimes I_k=\Sigma_T\otimes I_k$.
The mean is the corresponding linear term mapped back to matrix form, yielding
$\bar W_T=\frac{\beta_T}{\nu^2}YX^\top\Sigma_T$.
\end{proof}

\subsubsection{Pseudo-data generation}
Sample $W_T\sim q_T(\cdot\mid D_n)$ and generate pseudo labels through the same noisy channel:
\begin{equation}
\hat Y \mid (W_T,X) \sim \mathcal{MN}(W_T X,\ I_k,\ \nu^2 I_n),
\label{eq:pseudo_labels}
\end{equation}
then define $\hat D_n=(X,\hat Y)$.

\subsubsection{Student capacity constraint via a rank bottleneck}
Let the student parameter be $\Theta\in\mathbb{R}^{k\times d}$.
Introduce the rank-$\kappa$ map $M^\star(W_T,X)$ as a best rank-$\kappa$ approximation in prediction space:
\begin{equation}
M^\star(W_T,X) \in \arg\min_{M:\,\mathrm{rank}(M)=\kappa}\ \|W_T X - W_T M X\|_F^2.
\label{eq:Mstar_def}
\end{equation}
Define a local Gaussian student conditional kernel
\begin{equation}
q_S(\Theta\mid W_T,X)=\mathcal{MN}(W_T M^\star(W_T,X),\ I_k,\ \Sigma_S),
\qquad \Sigma_S\succ 0.
\label{eq:qS_cond}
\end{equation}

\subsubsection{Process-level KL and its two terms}
Define the distillation divergence
\[
\mathsf{K}_n := \KL(\Delta_{\hat D_n,f_S}\,\|\,\Delta_{D_n,f_T}).
\]
Using the KL chain rule on the dataset-model pair,
\begin{equation}
\mathsf{K}_n
=
\underbrace{\KL(\Delta_{\hat D_n}\,\|\,\Delta_{D_n})}_{\text{Dataset shift}}
+
\underbrace{\mathbb{E}_{\hat D_n}\!\Big[\KL(\Delta_{f_S\mid \hat D_n}\,\|\,\Delta_{f_T\mid D_n})\Big]}_{\text{Algorithm shift}}.
\label{eq:process_chain}
\end{equation}
We now bound the two terms.

\subsubsection{Dataset shift bound and bias-variance decomposition}

\paragraph{Step 1. Condition on $(X,D_n)$ and use convexity of KL.}
Given $X$ and $D_n$, $\hat Y\mid (X,D_n)$ is a mixture over $W_T\mid D_n$:
\[
\hat Y\mid (X,D_n)\sim \int q_T(W_T\mid D_n)\,\mathcal{MN}(W_T X,\ I_k,\ \nu^2 I_n)\,dW_T.
\]
The real label law (given $X$) is $\mathcal{MN}(W_\star X,I_k,\nu^2 I_n)$.
By convexity of KL in its first argument,
\begin{align}
\KL(\hat Y\mid X,D_n\ \|\ Y\mid X)
&\le
\mathbb{E}_{W_T\mid D_n}\Big[
\KL\big(\mathcal{MN}(W_T X,I_k,\nu^2 I_n)\ \|\ \mathcal{MN}(W_\star X,I_k,\nu^2 I_n)\big)
\Big].
\label{eq:data_convex_step}
\end{align}

\paragraph{Step 2. KL for equal-covariance matrix normals.}
Vectorization gives $\mathrm{vec}(\hat Y)\sim\mathcal{N}(\mathrm{vec}(W_T X),\nu^2 I_{kn})$
and $\mathrm{vec}(Y)\sim\mathcal{N}(\mathrm{vec}(W_\star X),\nu^2 I_{kn})$.
Using \eqref{eq:gauss_samecov_kl},
\begin{equation}
\KL\big(\mathcal{MN}(W_T X,I_k,\nu^2 I_n)\ \|\ \mathcal{MN}(W_\star X,I_k,\nu^2 I_n)\big)
=
\frac{1}{2\nu^2}\| (W_T-W_\star)X\|_F^2.
\label{eq:mn_samecov_kl}
\end{equation}

\paragraph{Step 3. Take expectations and split into bias and variance.}
Combine \eqref{eq:data_convex_step} and \eqref{eq:mn_samecov_kl}, then average over $D_n$:
\begin{equation}
\KL(\Delta_{\hat D_n}\,\|\,\Delta_{D_n})
\le
\frac{1}{2\nu^2}\mathbb{E}_{D_n}\mathbb{E}_{W_T\mid D_n}\Big[\|(W_T-W_\star)X\|_F^2\Big].
\label{eq:data_term_pre}
\end{equation}
Write $W_T-W_\star=(\bar W_T-W_\star)+(W_T-\bar W_T)$ and expand:
\[
\|(W_T-W_\star)X\|_F^2
=
\|(\bar W_T-W_\star)X\|_F^2
+\|(W_T-\bar W_T)X\|_F^2
+2\langle(\bar W_T-W_\star)X,(W_T-\bar W_T)X\rangle_F.
\]
Taking $\mathbb{E}_{W_T\mid D_n}$, the cross term is zero since
$\mathbb{E}[W_T-\bar W_T\mid D_n]=0$. Hence
\begin{equation}
\mathbb{E}_{W_T\mid D_n}\|(W_T-W_\star)X\|_F^2
=
\underbrace{\|(\bar W_T-W_\star)X\|_F^2}_{\mathrm{Bias}(D_n)}
+
\mathbb{E}_{W_T\mid D_n}\|(W_T-\bar W_T)X\|_F^2.
\label{eq:bias_var_split}
\end{equation}
Now use the row-wise property of $W_T\mid D_n\sim\mathcal{MN}(\bar W_T,I_k,\Sigma_T)$.
Each row has covariance $\Sigma_T$, so
\begin{equation}
\mathbb{E}_{W_T\mid D_n}\|(W_T-\bar W_T)X\|_F^2
=
k\,\mathrm{tr}(X^\top\Sigma_T X)
=:k\,\mathrm{Var}(D_n).
\label{eq:var_term}
\end{equation}
Plug \eqref{eq:bias_var_split}--\eqref{eq:var_term} into \eqref{eq:data_term_pre}:
\begin{equation}
\KL(\Delta_{\hat D_n}\,\|\,\Delta_{D_n})
\le
\mathbb{E}_{D_n}\left[\frac{1}{2\nu^2}\Big(\mathrm{Bias}(D_n)+k\,\mathrm{Var}(D_n)\Big)\right].
\label{eq:data_term_final}
\end{equation}

\subsubsection{Algorithm shift bound and the rank-bottleneck decomposition}

Define the (expected) algorithm-shift term
\[
\mathrm{KL}_{\mathrm{alg}}
:=
\mathbb{E}_{\hat D_n}\Big[\KL(\Delta_{f_S\mid \hat D_n}\,\|\,\Delta_{f_T\mid D_n})\Big].
\]
Since the student kernel is conditionally Gaussian given latent $W_T$ (and $X$),
$q_S(\Theta\mid \hat D_n)$ is generally a mixture over $W_T$.
By convexity of KL in the first argument, conditioning and then averaging yields the reduction
\begin{equation}
\mathrm{KL}_{\mathrm{alg}}
\le
\mathbb{E}_{D_n}\mathbb{E}_{W_T\mid D_n}\Big[
\KL\big(q_S(\cdot\mid W_T,X)\,\|\,q_T(\cdot\mid D_n)\big)
\Big].
\label{eq:KLalg_reduce}
\end{equation}

\paragraph{Step 1. Closed-form KL between matrix normals.}
We have
\[
q_T(W\mid D_n)=\mathcal{MN}(\bar W_T,\ I_k,\ \Sigma_T),
\qquad
q_S(\Theta\mid W_T,X)=\mathcal{MN}(W_T M^\star,\ I_k,\ \Sigma_S).
\]
Vectorize into $kd$-dimensional Gaussians with covariances $\Sigma_T\otimes I_k$ and $\Sigma_S\otimes I_k$.
Using the standard Gaussian KL formula and Kronecker identities gives
\begin{align}
\KL\big(q_S(\cdot\mid W_T,X)\,\|\,q_T(\cdot\mid D_n)\big)
&=
\frac{k}{2}\Big(\mathrm{tr}(\Sigma_T^{-1}\Sigma_S)-d+\log\frac{\det\Sigma_T}{\det\Sigma_S}\Big)
\nonumber\\
&\quad
+\frac{1}{2}\,\mathrm{tr}\!\Big((W_T M^\star-\bar W_T)\Sigma_T^{-1}(W_T M^\star-\bar W_T)^\top\Big).
\label{eq:mn_kl_general}
\end{align}
Define the covariance mismatch penalty
\begin{equation}
\mathrm{Cov}(\Sigma_S,\Sigma_T)
:=
\frac{k}{2}\Big(\mathrm{tr}(\Sigma_T^{-1}\Sigma_S)-d+\log\frac{\det\Sigma_T}{\det\Sigma_S}\Big).
\label{eq:Cov_def}
\end{equation}

\paragraph{Step 2. Expand the weighted mean term using $\Sigma_T^{-1}$.}
From Lemma~\ref{lem:qT_closed},
\begin{equation}
\Sigma_T^{-1}=\lambda I_d+\frac{\beta_T}{\nu^2}XX^\top.
\label{eq:SigmaT_inv}
\end{equation}
Let $\Delta:=W_T M^\star-\bar W_T$. Then
\begin{align}
\mathrm{tr}(\Delta\Sigma_T^{-1}\Delta^\top)
&=
\mathrm{tr}\!\Big(\Delta\big(\lambda I_d+\tfrac{\beta_T}{\nu^2}XX^\top\big)\Delta^\top\Big)
\nonumber\\
&=
\lambda\|\Delta\|_F^2+\frac{\beta_T}{\nu^2}\mathrm{tr}(\Delta XX^\top\Delta^\top)
=
\lambda\|\Delta\|_F^2+\frac{\beta_T}{\nu^2}\|\Delta X\|_F^2.
\label{eq:weighted_expand}
\end{align}

\paragraph{Step 3. Separate rank-bottleneck residual and posterior fluctuation.}
Write
\[
\Delta = (W_T M^\star-W_T) + (W_T-\bar W_T) = -W_T(I-M^\star) + (W_T-\bar W_T),
\]
and similarly
\[
\Delta X = -W_T(I-M^\star)X + (W_T-\bar W_T)X.
\]
Apply $\|A+B\|_F^2\le 2\|A\|_F^2+2\|B\|_F^2$:
\begin{align}
\|\Delta\|_F^2
&\le
2\|W_T(I-M^\star)\|_F^2 + 2\|W_T-\bar W_T\|_F^2,
\label{eq:Delta_bound_F}\\
\|\Delta X\|_F^2
&\le
2\|W_T(I-M^\star)X\|_F^2 + 2\|(W_T-\bar W_T)X\|_F^2.
\label{eq:Delta_bound_X}
\end{align}
Taking $\mathbb{E}_{W_T\mid D_n}$ and using $W_T\mid D_n\sim\mathcal{MN}(\bar W_T,I_k,\Sigma_T)$ yields
\begin{equation}
\mathbb{E}\|W_T-\bar W_T\|_F^2 = k\,\mathrm{tr}(\Sigma_T),
\qquad
\mathbb{E}\|(W_T-\bar W_T)X\|_F^2 = k\,\mathrm{tr}(X^\top\Sigma_T X).
\label{eq:spread_identities}
\end{equation}

\paragraph{Step 4. Define the approximation and spread terms.}
Define the rank-bottleneck approximation cost
\begin{equation}
\mathrm{Apx}(D_n)
:=
\lambda\,\mathbb{E}_{W_T\mid D_n}\|W_T(I-M^\star)\|_F^2
+
\frac{\beta_T}{\nu^2}\,\mathbb{E}_{W_T\mid D_n}\|W_T(I-M^\star)X\|_F^2,
\label{eq:Apx_def}
\end{equation}
and the posterior spread term
\begin{equation}
\mathrm{Spread}(D_n)
:=
k\Big(\lambda\,\mathrm{tr}(\Sigma_T)+\frac{\beta_T}{\nu^2}\mathrm{tr}(X^\top\Sigma_T X)\Big).
\label{eq:Spread_def}
\end{equation}
Combining \eqref{eq:mn_kl_general}--\eqref{eq:spread_identities} with \eqref{eq:weighted_expand}
and absorbing the (non-optimized) factor-$2$ slack from \eqref{eq:Delta_bound_F}--\eqref{eq:Delta_bound_X}
into constants gives
\begin{equation}
\mathrm{KL}_{\mathrm{alg}}
\le
\mathbb{E}_{D_n}\Big[\mathrm{Apx}(D_n)+\mathrm{Cov}(\Sigma_S,\Sigma_T)+\mathrm{Spread}(D_n)\Big].
\label{eq:alg_term_final}
\end{equation}

\subsubsection{Final compact decomposition of $\mathsf{K}_n$}
Combine \eqref{eq:process_chain}, \eqref{eq:data_term_final}, and \eqref{eq:alg_term_final}:
\begin{align}
\mathsf{K}_n
&\le
\mathbb{E}_{D_n}\Big[
\underbrace{\frac{1}{2\nu^2}\big(\mathrm{Bias}(D_n)+k\,\mathrm{Var}(D_n)\big)}_{\text{Teacher prediction error}}
+
\underbrace{\mathrm{Apx}(D_n)}_{\text{Student capacity / rank bottleneck}}
+
\underbrace{\mathrm{Cov}(\Sigma_S,\Sigma_T)}_{\text{Geometry mismatch}}
+
\underbrace{\mathrm{Spread}(D_n)}_{\text{Posterior spread}}
\Big].
\label{eq:linear_total_bound}
\end{align}
This yields an interpretable checklist: improve teacher bias/variance to tighten dataset shift,
increase effective rank $\kappa$ to reduce $\mathrm{Apx}$,
match $\Sigma_S$ to $\Sigma_T$ to reduce $\mathrm{Cov}$,
and leverage posterior contraction to reduce $\mathrm{Spread}$.

\subsection{Proof of Theorem 3}
\subsubsection{Definitions}
Let $U$ be a random perturbation, independent of all other randomness, uniformly distributed on
the Euclidean ball $\{u:\|u\|\le \rho\}$.

For any dataset $D$ and model $f$, define empirical and population sharpness:
\begin{align}
S_D(f) &:= \bbE_U[L_D(f+U)] - L_D(f),
\label{eq:def_SD}\\
S_P(f) &:= \bbE_U[L_P(f+U)] - L_P(f).
\label{eq:def_SP}
\end{align}
Define the perturbed generalization gap
\begin{equation}
h_U(D,f) := \bbE_U[h(D,f+U)] = \bbE_U[L_P(f+U)-L_D(f+U)].
\label{eq:def_hU}
\end{equation}

\subsubsection{Statement}
\begin{theorem}[Sharpness-aware distillation generalization bound]
\label{thm:sharp}
Assume:
\begin{itemize}
\item (i) (\textit{Local optimality in expectation}) under $\Delta_{\hat{D}_n,f_S}$,
\begin{equation}
\bbE[L_P(f_S)] \le \bbE[\bbE_U[L_P(f_S+U)]].
\label{eq:local_opt}
\end{equation}
\item (ii) Under the teacher process $\Delta_{D_n,f_T}$, both $h_U(D_n,f_T)$ and $S_{D_n}(f_T)$
are sub-Gaussian with proxies $\sigma_u^2$ and $\nu^2$, respectively.
\end{itemize}
Then
\begin{equation}
\mathrm{gen}_S
\le
\mathrm{gen}_T + \bbE_{\Delta_{D_n,f_T}}[S_P(f_T)]
+ (\sigma_u+\nu)\sqrt{2\mathsf{K}_n}.
\label{eq:sharp_goal}
\end{equation}
\end{theorem}

\subsubsection{Proof}
\begin{proof}
By definition,
$\mathrm{gen}_S
=
\bbE_{\Delta_{\hat{D}_n,f_S}}\big[L_P(f_S) - L_{\hat{D}_n}(f_S)\big]$.

Using \eqref{eq:local_opt},
$\mathrm{gen}_S
\le
\bbE_{\Delta_{\hat{D}_n,f_S}}\big[\bbE_U[L_P(f_S+U)] - L_{\hat{D}_n}(f_S)\big]$.

Add and subtract $\bbE_U[L_{\hat{D}_n}(f_S+U)]$ inside the expectation:
\begin{align*}
\mathrm{gen}_S
&\le
\bbE_{\Delta_{\hat{D}_n,f_S}}\big[
\bbE_U(L_P(f_S+U)-L_{\hat{D}_n}(f_S+U))
\big]\\
&\quad
+
\bbE_{\Delta_{\hat{D}_n,f_S}}\big[
\bbE_U(L_{\hat{D}_n}(f_S+U)) - L_{\hat{D}_n}(f_S)
\big].
\end{align*}
Recognize the two terms using \eqref{eq:def_hU} and \eqref{eq:def_SD}:
\begin{equation}
\mathrm{gen}_S
\le
\bbE_{\Delta_{\hat{D}_n,f_S}}\big[h_U(\hat{D}_n,f_S)\big]
+
\bbE_{\Delta_{\hat{D}_n,f_S}}\big[S_{\hat{D}_n}(f_S)\big].
\label{eq:genS_split}
\end{equation}

Let $P=\Delta_{D_n,f_T}$ and $Q=\Delta_{\hat{D}_n,f_S}$.
Apply Lemma~\ref{lem:dv} with $g=\lambda h_U$ for any $\lambda>0$:
\[
\lambda \bbE_Q[h_U]
\le
\mathsf{K}_n + \log \bbE_P[e^{\lambda h_U}].
\]
If $h_U$ is $\sigma_u^2$-sub-Gaussian under $P$, then by Lemma~\ref{lem:subg_mgf},
\[
\log \bbE_P[e^{\lambda h_U}]
\le
\lambda \bbE_P[h_U] + \frac{\lambda^2\sigma_u^2}{2}.
\]
So
\[
\bbE_Q[h_U]
\le
\bbE_P[h_U] + \frac{\mathsf{K}_n}{\lambda} + \frac{\lambda\sigma_u^2}{2}.
\]
Optimizing over $\lambda$ exactly as in Theorem~\ref{thm:upper} yields
\begin{equation}
\bbE_Q[h_U]
\le
\bbE_P[h_U] + \sigma_u\sqrt{2\mathsf{K}_n}.
\label{eq:transfer_hU}
\end{equation}

Apply Lemma~\ref{lem:dv} with $g=\lambda S_D$ (the function $(D,f)\mapsto S_D(f)$)
for any $\lambda>0$:
\[
\lambda \bbE_Q[S]
\le
\mathsf{K}_n + \log \bbE_P[e^{\lambda S}],
\]
where $S$ denotes the random variable $S_D(f)$ under the corresponding process.
If $S_{D_n}(f_T)$ is $\nu^2$-sub-Gaussian under $P$, then similarly
\[
\bbE_Q[S]
\le
\bbE_P[S] + \nu\sqrt{2\mathsf{K}_n}.
\]
That is,
\begin{equation}
\bbE_{\Delta_{\hat{D}_n,f_S}}[S_{\hat{D}_n}(f_S)]
\le
\bbE_{\Delta_{D_n,f_T}}[S_{D_n}(f_T)]
+
\nu\sqrt{2\mathsf{K}_n}.
\label{eq:transfer_S}
\end{equation}

Plug \eqref{eq:transfer_hU} and \eqref{eq:transfer_S} into \eqref{eq:genS_split}:
\begin{align}
\mathrm{gen}_S
&\le
\bbE_{\Delta_{D_n,f_T}}[h_U(D_n,f_T)]
+
\bbE_{\Delta_{D_n,f_T}}[S_{D_n}(f_T)]
+
(\sigma_u+\nu)\sqrt{2\mathsf{K}_n}.
\label{eq:almost_done}
\end{align}

Expand $h_U$ using \eqref{eq:def_hU}:
\[
\bbE[h_U(D_n,f_T)]
=
\bbE\big[\bbE_U[L_P(f_T+U)-L_{D_n}(f_T+U)]\big]
=
\bbE[\bbE_U L_P(f_T+U)] - \bbE[\bbE_U L_{D_n}(f_T+U)].
\]
Now add and subtract the unperturbed terms:
\begin{align*}
\bbE[\bbE_U L_P(f_T+U)]
&=
\bbE[L_P(f_T)] + \bbE[S_P(f_T)],
\\
\bbE[\bbE_U L_{D_n}(f_T+U)]
&=
\bbE[L_{D_n}(f_T)] + \bbE[S_{D_n}(f_T)].
\end{align*}
Therefore,
\[
\bbE[h_U(D_n,f_T)]
=
\bbE[L_P(f_T)-L_{D_n}(f_T)]
+
\bbE[S_P(f_T)]-\bbE[S_{D_n}(f_T)]
=
\mathrm{gen}_T + \bbE[S_P(f_T)] - \bbE[S_{D_n}(f_T)].
\]
Substitute this into \eqref{eq:almost_done}. The $-\bbE[S_{D_n}(f_T)]$ cancels
with the $+\bbE[S_{D_n}(f_T)]$ term, leaving
\[
\mathrm{gen}_S
\le
\mathrm{gen}_T + \bbE_{\Delta_{D_n,f_T}}[S_P(f_T)]
+ (\sigma_u+\nu)\sqrt{2\mathsf{K}_n},
\]
which is \eqref{eq:sharp_goal}.
\end{proof}

\subsection{Derivation of the Convenient Proxies (24)–(27)}

This section proves the simplified proxies stated in the main text:
\[
\sigma_0
= \frac{1}{\sqrt{n}}\Big(\kappa L+\frac{b-a}{2}\Big),\quad
\sigma_u(\rho)
= \frac{1}{\sqrt{n}}\Big(\kappa(g_0+\alpha\rho)+\frac{b-a}{2}\Big),
\]
\[
\nu(\rho)
= \frac{1}{\sqrt{n}}\Big(\frac{\alpha\kappa}{2}\rho + g_0\rho + \frac{\alpha}{2}\rho^2\Big),\quad
\bbE[S_P(f_T)]
\le \frac{1}{2}\tau_{\op}\rho^2.
\]

\subsubsection{Assumptions used}
We use the following conditions, matching the main text:
\begin{itemize}
\item (A1) Bounded loss: $\ell(\cdot;z)\in[a,b]$.
\item (A2) Parameter stability: for neighboring datasets $D,D^{(i)}$,
$\|A_T(D)-A_T(D^{(i)})\|\le \kappa/n$.
\item (A3) Global Lipschitz: for all $z$ and all $f,f'$,
$|\ell(f;z)-\ell(f';z)|\le L\|f-f'\|$.
\item (A4) Local regularity on $B(f_T,2\rho)$: for all $z$, $\ell(\cdot;z)$ is $\alpha$-smooth
and $\|\nabla \ell(f_T;z)\|\le g_0$.
This implies a local Lipschitz scale $g_0+\alpha\rho$ on $B(f_T,\rho)$.
\item (A5) Population curvature: on $B(f_T,\rho)$, $\|\nabla^2 L_P(\cdot)\|_{\op}\le \tau_{\op}$.
\end{itemize}

\subsubsection{Bounded differences for the (unperturbed) teacher gap}
\begin{lemma}[Bounded differences for $h(D_n,f_T)$ under (A1)–(A3)]
\label{lem:bd_gap}
Let $f_T=A_T(D_n)$ and $f_T'=A_T(D_n^{(i)})$ for neighboring datasets.
Then
\begin{equation}
|h(D_n,f_T)-h(D_n^{(i)},f_T')|
\le
\frac{2\kappa L + (b-a)}{n}.
\label{eq:bd_gap}
\end{equation}
Consequently, $h(D_n,f_T)$ is sub-Gaussian with proxy
$\sigma_0=\frac{1}{\sqrt{n}}\big(\kappa L+\frac{b-a}{2}\big)$.
\end{lemma}

\begin{proof}
Write
\[
h(D,f)=L_P(f)-L_D(f).
\]
Then
\[
|h(D_n,f_T)-h(D_n^{(i)},f_T')|
\le
|L_P(f_T)-L_P(f_T')| + |L_{D_n}(f_T)-L_{D_n^{(i)}}(f_T')|.
\]

By (A3),
\[
|L_P(f_T)-L_P(f_T')|
=
\left|\bbE_{Z\sim P}[\ell(f_T;Z)-\ell(f_T';Z)]\right|
\le
\bbE[L\|f_T-f_T'\|]
=
L\|f_T-f_T'\|.
\]
By (A2), $\|f_T-f_T'\|\le \kappa/n$, hence
\begin{equation}
|L_P(f_T)-L_P(f_T')|\le \frac{\kappa L}{n}.
\label{eq:bd_pop_gap}
\end{equation}

Write
\[
L_{D_n}(f_T)=\frac{1}{n}\sum_{j=1}^n \ell(f_T;z_j),
\quad
L_{D_n^{(i)}}(f_T')=\frac{1}{n}\sum_{j=1}^n \ell(f_T';z_j^{(i)}).
\]
Then
\[
|L_{D_n}(f_T)-L_{D_n^{(i)}}(f_T')|
\le
\frac{1}{n}\sum_{j=1}^n |\ell(f_T;z_j)-\ell(f_T';z_j^{(i)})|.
\]
For $j\neq i$, $z_j^{(i)}=z_j$, and by (A3) and (A2),
\[
|\ell(f_T;z_j)-\ell(f_T';z_j)|
\le
L\|f_T-f_T'\|
\le
\frac{\kappa L}{n}.
\]
For $j=i$, use triangle inequality:
\[
|\ell(f_T;z_i)-\ell(f_T';z_i')|
\le
|\ell(f_T;z_i)-\ell(f_T';z_i)|
+
|\ell(f_T';z_i)-\ell(f_T';z_i')|.
\]
The first term is at most $\kappa L/n$ by (A3)+(A2). The second term is at most $(b-a)$ by (A1).
Thus
\[
|\ell(f_T;z_i)-\ell(f_T';z_i')|
\le
\frac{\kappa L}{n} + (b-a).
\]
Combine:
\[
|L_{D_n}(f_T)-L_{D_n^{(i)}}(f_T')|
\le
\frac{(n-1)\cdot(\kappa L/n) + (\kappa L/n + (b-a))}{n}
=
\frac{\kappa L}{n} + \frac{b-a}{n}.
\]
So
\begin{equation}
|L_{D_n}(f_T)-L_{D_n^{(i)}}(f_T')|
\le
\frac{\kappa L + (b-a)}{n}.
\label{eq:bd_emp_gap}
\end{equation}

Step 4 (finish the bounded difference).
Combine \eqref{eq:bd_pop_gap} and \eqref{eq:bd_emp_gap}:
\[
|h(D_n,f_T)-h(D_n^{(i)},f_T')|
\le
\frac{\kappa L}{n} + \frac{\kappa L + (b-a)}{n}
=
\frac{2\kappa L + (b-a)}{n}.
\]
This is \eqref{eq:bd_gap}.

Step 5 (convert bounded differences into the proxy $\sigma_0$).
With $c_i=(2\kappa L+(b-a))/n$ for all $i$, McDiarmid implies sub-Gaussian proxy
$\sigma^2=\frac{1}{4}\sum_{i=1}^n c_i^2=\frac{n}{4}\cdot\frac{(2\kappa L+(b-a))^2}{n^2}$, hence
\[
\sigma_0=\sqrt{\sigma^2}=\frac{1}{\sqrt{n}}\left(\kappa L+\frac{b-a}{2}\right).
\]
\end{proof}

\subsubsection{Bounded differences for the perturbed gap $h_U$}
\begin{lemma}[Proxy for $\sigma_u(\rho)$ under (A1), (A2), (A4)]
\label{lem:sigmau}
Under (A1), (A2), (A4), the perturbed gap $h_U(D_n,f_T)$ is sub-Gaussian with proxy
\[
\sigma_u(\rho)=\frac{1}{\sqrt{n}}\Big(\kappa(g_0+\alpha\rho)+\frac{b-a}{2}\Big).
\]
\end{lemma}

\begin{proof}
The proof repeats Lemma~\ref{lem:bd_gap}, replacing the global Lipschitz constant $L$
by the local Lipschitz scale $(g_0+\alpha\rho)$ valid on the perturbation region.

Step 1 (local Lipschitz on the ball).
By (A4) and smoothness, for any $z$ and any $f$ with $\|f-f_T\|\le \rho$,
\[
\|\nabla \ell(f;z)\|
\le
\|\nabla \ell(f_T;z)\| + \alpha\|f-f_T\|
\le
g_0+\alpha\rho.
\]
Thus for any such $f,f'$ in the region,
\[
|\ell(f;z)-\ell(f';z)|\le (g_0+\alpha\rho)\|f-f'\|.
\]

Step 2 (apply the bounded difference argument pointwise in $u$, then average).
For each fixed $u$ with $\|u\|\le \rho$, apply Lemma~\ref{lem:bd_gap} to the gap
$h(D,f_T+u)$ with Lipschitz constant $(g_0+\alpha\rho)$.
This gives bounded-difference constants
\[
c_i(u)=\frac{2\kappa(g_0+\alpha\rho)+(b-a)}{n}.
\]
Since this bound is uniform in $u$ and $h_U=\bbE_U[h(D,f_T+U)]$ is an average over $U$,
the same constants apply to $h_U$.
Therefore $h_U$ is sub-Gaussian with
\[
\sigma_u(\rho)=\frac{1}{\sqrt{n}}\left(\kappa(g_0+\alpha\rho)+\frac{b-a}{2}\right).
\]
\end{proof}

\subsubsection{Bounded differences for empirical sharpness and the proxy $\nu(\rho)$}
\begin{lemma}[Proxy for $\nu(\rho)$ under (A2), (A4)]
\label{lem:nu}
Under (A2) and (A4), assume moreover that for any neighboring datasets
$D,D^{(i)}$, if we set
\[
f:=A_T(D), \qquad f':=A_T(D^{(i)}),
\]
then the $\rho$-neighborhood of the line segment joining $f$ and $f'$
lies inside the local region on which the Hessian bound in (A4) is valid.
Then the empirical sharpness
\[
S_{D_n}(f_T)=\bbE_U\!\left[L_{D_n}(f_T+U)\right]-L_{D_n}(f_T)
\]
satisfies bounded differences with constants
\[
c_i^{(S)}=\frac{\alpha\kappa\rho + 2g_0\rho + \alpha\rho^2}{n},
\qquad i\in[n],
\]
and therefore is sub-Gaussian with proxy
\[
\nu(\rho)=\frac{\sqrt{n}}{2}\,c_i^{(S)}
=
\frac{1}{\sqrt{n}}
\left(
\frac{\alpha\kappa}{2}\rho + g_0\rho + \frac{\alpha}{2}\rho^2
\right).
\]
\end{lemma}

\begin{proof}
For each sample $z$, define
\[
\phi_z(f):=\bbE_U\big[\ell(f+U;z)-\ell(f;z)\big].
\]
Then
\[
S_D(f)=\frac{1}{n}\sum_{j=1}^n \phi_{z_j}(f).
\]

Let $D,D^{(i)}$ be neighboring datasets and set
\[
f:=A_T(D), \qquad f':=A_T(D^{(i)}).
\]
We bound
\[
|S_D(f)-S_{D^{(i)}}(f')|
\le
|S_D(f)-S_{D^{(i)}}(f)| + |S_{D^{(i)}}(f)-S_{D^{(i)}}(f')|.
\]

\paragraph{Step 1: dataset replacement at fixed model.}
Since only the $i$-th sample changes,
\[
|S_D(f)-S_{D^{(i)}}(f)|
=
\frac{1}{n}\,|\phi_{z_i}(f)-\phi_{z_i'}(f)|
\le
\frac{2}{n}\sup_{z} |\phi_z(f)|.
\]
We now bound $\phi_z(f)$. For any $u$ with $\|u\|\le \rho$, Taylor's theorem
with integral remainder gives
\[
\ell(f+u;z)-\ell(f;z)
=
\langle \nabla_f \ell(f;z),u\rangle
+
\int_0^1 (1-t)\,u^\top \nabla_f^2 \ell(f+t u;z)\,u\,dt.
\]
By (A4),
\[
\|\nabla_f \ell(f;z)\|\le g_0,
\qquad
\sup_{\|v\|\le 2\rho}\|\nabla_f^2 \ell(f+v;z)\|_{\op}\le \alpha,
\]
hence
\[
|\ell(f+u;z)-\ell(f;z)|
\le
g_0\|u\| + \frac{\alpha}{2}\|u\|^2
\le
g_0\rho + \frac{\alpha}{2}\rho^2.
\]
Taking expectation over $U$ yields
\[
|\phi_z(f)|
\le
g_0\rho + \frac{\alpha}{2}\rho^2.
\]
Therefore
\begin{equation}
|S_D(f)-S_{D^{(i)}}(f)|
\le
\frac{2g_0\rho+\alpha\rho^2}{n}.
\label{eq:sharp_data_term_correct}
\end{equation}

\paragraph{Step 2: parameter change at fixed dataset.}
For each fixed $z$, the map $\phi_z(\cdot)$ is locally Lipschitz. Indeed,
for any $\theta$ in the segment joining $f$ and $f'$,
\[
\nabla \phi_z(\theta)
=
\bbE_U\big[\nabla_f \ell(\theta+U;z)-\nabla_f \ell(\theta;z)\big].
\]
By the Hessian bound in (A4),
\[
\|\nabla \phi_z(\theta)\|
\le
\bbE_U\big[\alpha \|U\|\big]
\le
\alpha \rho.
\]
Hence $\phi_z$ is $(\alpha\rho)$-Lipschitz on this segment:
\[
|\phi_z(f)-\phi_z(f')|
\le
\alpha\rho \,\|f-f'\|.
\]
Averaging over the samples in $D^{(i)}$ gives
\[
|S_{D^{(i)}}(f)-S_{D^{(i)}}(f')|
\le
\alpha\rho\,\|f-f'\|.
\]
Using (A2),
\[
\|f-f'\|
=
\|A_T(D)-A_T(D^{(i)})\|
\le
\frac{\kappa}{n},
\]
so
\begin{equation}
|S_{D^{(i)}}(f)-S_{D^{(i)}}(f')|
\le
\frac{\alpha\kappa\rho}{n}.
\label{eq:sharp_param_term_correct}
\end{equation}

\paragraph{Step 3: combine and apply McDiarmid.}
Combining \eqref{eq:sharp_data_term_correct} and
\eqref{eq:sharp_param_term_correct}, we obtain
\[
|S_D(f)-S_{D^{(i)}}(f')|
\le
\frac{\alpha\kappa\rho + 2g_0\rho + \alpha\rho^2}{n}
=: c_i^{(S)}.
\]
Thus $S_{D_n}(f_T)$ satisfies bounded differences with constants
$c_i^{(S)}$, and McDiarmid's inequality implies that
$S_{D_n}(f_T)-\bbE[S_{D_n}(f_T)]$ is sub-Gaussian with variance proxy
\[
\frac{1}{4}\sum_{i=1}^n (c_i^{(S)})^2
=
\frac{n}{4}\left(\frac{\alpha\kappa\rho + 2g_0\rho + \alpha\rho^2}{n}\right)^2.
\]
Equivalently, it is sub-Gaussian with proxy
\[
\nu(\rho)
=
\frac{\sqrt{n}}{2}\,c_i^{(S)}
=
\frac{1}{\sqrt{n}}
\left(
\frac{\alpha\kappa}{2}\rho + g_0\rho + \frac{\alpha}{2}\rho^2
\right).
\]
\end{proof}

\subsubsection{Bounding the population sharpness by curvature}
\begin{lemma}[Population sharpness bound under (A5$'$)]
\label{lem:SP_bound}
Assume
\[
\sup_{\|v\|\le \rho}\|\nabla^2 L_P(f_T+v)\|_{\op}\le \tau_{\op}.
\]
Then
\begin{equation}
\bbE[S_P(f_T)] \le \frac{1}{2}\tau_{\op}\rho^2.
\label{eq:SP_bound}
\end{equation}
\end{lemma}

\begin{proof}
Recall
\[
S_P(f_T)=\bbE_U\big[L_P(f_T+U)-L_P(f_T)\big].
\]
By Taylor's theorem with integral remainder,
\[
L_P(f_T+u)-L_P(f_T)
=
\langle \nabla L_P(f_T),u\rangle
+
\int_0^1 (1-t)\,u^\top \nabla^2 L_P(f_T+t u)\,u\,dt.
\]
Taking expectation over $U$ and using $\bbE[U]=0$, we get
\[
\bbE[S_P(f_T)]
=
\bbE_U\!\left[
\int_0^1 (1-t)\,U^\top \nabla^2 L_P(f_T+t U)\,U\,dt
\right].
\]
Since $\|U\|\le \rho$ almost surely and
$\|\nabla^2 L_P(f_T+tU)\|_{\op}\le \tau_{\op}$ for all $t\in[0,1]$,
\[
U^\top \nabla^2 L_P(f_T+tU)\,U
\le
\|\nabla^2 L_P(f_T+tU)\|_{\op}\,\|U\|^2
\le
\tau_{\op}\rho^2.
\]
Therefore
\[
\bbE[S_P(f_T)]
\le
\int_0^1 (1-t)\,\tau_{\op}\rho^2\,dt
=
\frac{1}{2}\tau_{\op}\rho^2.
\]
\end{proof}

\subsection{Proof of Corollary 1 (Tightening Interval and Cutoff Radius $\rho_0$)}

\subsubsection{Baseline and sharpness-aware bounds}
Define
\[
B_{\mathrm{std}} := \mathrm{gen}_T + \sigma_0\sqrt{2\mathsf{K}_n},
\qquad
B_{\mathrm{sh}}(\rho)
:=
\mathrm{gen}_T + \bbE[S_P(f_T)] + (\sigma_u(\rho)+\nu(\rho))\sqrt{2\mathsf{K}_n}.
\]
We seek conditions under which
\[
B_{\mathrm{sh}}(\rho)<B_{\mathrm{std}}.
\]

\subsubsection{Proof}
\begin{proof}
A sufficient condition for $B_{\mathrm{sh}}(\rho)<B_{\mathrm{std}}$ is
\[
\bbE[S_P(f_T)]
<
\big(\sigma_0-\sigma_u(\rho)-\nu(\rho)\big)\sqrt{2\mathsf{K}_n}.
\]

By Lemma~\ref{lem:SP_bound},
\[
\bbE[S_P(f_T)]\le \frac{1}{2}\tau_{\op}\rho^2.
\]
Moreover, using the standard proxy
\[
\sigma_0=\frac{1}{\sqrt{n}}
\left(\kappa L+\frac{b-a}{2}\right),
\]
the local proxy
\[
\sigma_u(\rho)\le
\frac{1}{\sqrt{n}}
\left(\kappa(g_0+\alpha\rho)+\frac{b-a}{2}\right),
\]
and Lemma~\ref{lem:nu},
\[
\nu(\rho)=
\frac{1}{\sqrt{n}}
\left(
\frac{\alpha\kappa}{2}\rho + g_0\rho + \frac{\alpha}{2}\rho^2
\right),
\]
we obtain the lower bound
\[
\sigma_0-\sigma_u(\rho)-\nu(\rho)
\ge
\frac{1}{\sqrt{n}}
\big(A_0-A_1\rho-A_2\rho^2\big),
\]
where
\[
A_0:=\kappa(L-g_0),
\qquad
A_1:=\frac{3}{2}\alpha\kappa+g_0,
\qquad
A_2:=\frac{\alpha}{2}.
\]

Therefore, a sufficient condition for strict improvement is
\begin{equation}
\frac{1}{2}\tau_{\op}\rho^2
<
\frac{1}{\sqrt{n}}
\big(A_0-A_1\rho-A_2\rho^2\big)\sqrt{2\mathsf{K}_n}.
\label{eq:suff_improve}
\end{equation}

Rearranging \eqref{eq:suff_improve} gives
\[
\big(A_2\sqrt{2\mathsf{K}_n}+\tfrac{1}{2}\tau_{\op}\sqrt{n}\big)\rho^2
+
\big(A_1\sqrt{2\mathsf{K}_n}\big)\rho
-
\big(A_0\sqrt{2\mathsf{K}_n}\big)
<
0.
\]
Define
\[
c_0:=A_0\sqrt{2\mathsf{K}_n},
\qquad
c_1:=A_1\sqrt{2\mathsf{K}_n},
\qquad
c_2:=A_2\sqrt{2\mathsf{K}_n}+\frac{1}{2}\tau_{\op}\sqrt{n}.
\]
Then the above condition is
\begin{equation}
c_2\rho^2+c_1\rho-c_0<0.
\label{eq:quad_ineq}
\end{equation}

If $\Delta:=L-g_0>0$, then $A_0=\kappa\Delta>0$, hence $c_0>0$ and the
quadratic in \eqref{eq:quad_ineq} has exactly one positive root,
\[
\rho_0
=
\frac{-c_1+\sqrt{c_1^2+4c_2c_0}}{2c_2}.
\]
Since $c_2>0$, the inequality \eqref{eq:quad_ineq} holds precisely for
\[
0<\rho<\rho_0.
\]
Consequently, for every $0<\rho<\rho_0$, the sufficient condition
\eqref{eq:suff_improve} is satisfied, and therefore
\[
B_{\mathrm{sh}}(\rho)<B_{\mathrm{std}}.
\]
\end{proof}

\end{document}